\documentclass{article}
\usepackage[utf8]{inputenc}
\usepackage[T1]{fontenc}
\usepackage{amsmath}
\usepackage{dsfont}
\usepackage{fancyhdr}
\usepackage[table,xcdraw]{xcolor}
\usepackage{graphicx}
\usepackage{multirow}
\usepackage{endnotes}
\usepackage{float}
\usepackage[nottoc]{tocbibind}
\usepackage{vmargin}
\usepackage{amssymb}
\usepackage{subfigure}
\usepackage{import}
\usepackage{wrapfig}
\usepackage{amsthm}
\usepackage{rotating}
\usepackage[title]{appendix}
\usepackage{listings}

\usepackage{soul}

\usepackage{stmaryrd}
\usepackage{tikz}
\usetikzlibrary{arrows.meta, positioning}

\usepackage{mleftright}
\mleftright
\usetikzlibrary{matrix}
\usepackage{tikz}
\usetikzlibrary{shapes,arrows,chains, decorations.pathmorphing, decorations.pathreplacing}
\tikzset{quantum/.style={decorate, decoration=snake}}

\newcommand{\ket}[1]{|#1\rangle}

\newcommand{\ketbra}[2]{|#1\rangle\langle#2|}

\newcommand{\abs}[1]{\lvert #1\rvert}

\newcommand{\QPVBB}{$\mathrm{QPV}_{\mathrm{BB84}}$}

\newcommand{\QPVBBf}{$\mathrm{QPV}_{\mathrm{BB84}}^{f}$}
\newcommand{\cQPVBBf}{$\textsf{c}$-$\mathrm{QPV}_{\mathrm{BB84}}^{f}$}

\newcommand{\QPVBBfparallel}{$\mathrm{QPV}_{\mathrm{BB84}}^{f: n\rightarrow m}$}

\newcommand{\cQPVBBfparallel}{$\mathsf{c}$-$\mathrm{QPV}_{\mathrm{BB84}}^{f: n\rightarrow m}$}

\newcommand{\tr}[1]{\mathrm{Tr}\left[#1\right]}
\newcommand{\ptr}[2]{\mathrm{Tr}_{#1}\left[#2\right]}
\newcommand{\Tr}{\operatorname{Tr}}

\usetikzlibrary{decorations.pathreplacing,calligraphy}
\usetikzlibrary{matrix}
\usepackage{tikz}
\usetikzlibrary{shapes,arrows,chains, decorations.pathmorphing, decorations.pathreplacing}
\tikzset{quantum/.style={decorate, decoration=snake}}

\newcommand{\diagdots}[3][-25]{%
  \rotatebox{#1}{\makebox[0pt]{\makebox[#2]{\xleaders\hbox{$\cdot$\hskip#3}\hfill\kern0pt}}}%
}

\usepackage{slashed}
\usepackage{enumitem}

\usepackage{authblk}

\usepackage{tcolorbox} 
\usepackage{array}     
\usepackage{colortbl}  
\usepackage{lipsum}

\usepackage{MnSymbol}

\definecolor{secondaryColor}{HTML}{5869bc}
\usepackage{url}
\usepackage[normalem]{ulem}
\usepackage{hyperref}
\hypersetup{
  colorlinks=true,
  allcolors=secondaryColor
}

\usepackage{amsthm}

\newtheorem{theorem}{Theorem}[section]
\newtheorem{lemma}[theorem]{Lemma}

\newcommand{\acc}{\mathrm{accept}}

\usepackage{cleveref}

\crefname{figure}{Fig.}{Figs.}
\crefname{table}{Table}{Tables}
\crefname{section}{Section}{Sections}

\crefname{theorem}{Theorem}{Theorems}
\crefname{lemma}{Lemma}{Lemmas}
\crefname{prop}{Proposition}{Propositions}
\crefname{corollary}{Corollary}{Corollaries}
\crefname{observation}{Observation}{Observations}
\crefname{remark}{Remark}{Remarks}
\crefname{definition}{Definition}{Definitions}
\crefname{ex}{Example}{Examples}
\crefname{result}{Result}{Results}
\crefname{attack}{Attack}{Attacks}

\usepackage[table,xcdraw]{xcolor}
\usepackage{amssymb}
\usepackage{pifont}
\usetikzlibrary{calc}
\newcommand{\Verifzero}{\textsf{V}$_0$}
\newcommand{\Verifone}{\textsf{V}$_1$}
\newcommand{\prover}{\textsf{P}}

\usepackage{tabularx}

\begin{document}
\title{Arbitrarily Loss-Tolerant Quantum Position Verification in a Single Execution}

\author[1]{Lloren\c{c} Escol\`a-Farr\`as}
\author[1]{{Boris \v{S}kori\'{c}}}
\author[2,3]{Florian Speelman}
\newcommand{\fs}[1]{{\textcolor{red}{[Florian: #1]}}}
\newcommand{\bs}[1]{{\textcolor{purple}{[Boris: #1]}}}

\affil[1]{Technische Universiteit Eindhoven, Eindhoven, The Netherlands }
\affil[2]{QuSoft, CWI Amsterdam, The Netherlands}
\affil[3]{University of Amsterdam, Amsterdam, The Netherlands}

\renewcommand\Affilfont{\itshape\small}
\maketitle
\vspace{-1em}
\begin{abstract}
Quantum position verification (QPV) aims to certify the location of an untrusted prover, but faces two major obstacles: fundamentally, entanglement-based attacks and, experimentally, photon loss. A commitment-based modification introduced in \textit{Phys.\ Rev.\ Lett.}\ \textbf{135}, 260801 addresses both in sequentially repeated protocols. Its security analysis, however, relies on the sequential structure and does not extend to parallel repetition. We use a different proof approach to establish security of the commitment modification for a parallel BB84-based QPV protocol. Against bounded-entanglement adversaries, the acceptance probability decays exponentially in the minimum number~\(k\) of successfully committed qubits. The protocol tolerates arbitrary transmission loss and noise rates up to~\(3.7\%\), while allowing arbitrarily slow quantum communication. This yields a fully loss-tolerant, single-execution QPV protocol secure against bounded-entanglement attacks, removing transmission loss as a fundamental limitation on verification distance. We also revisit sequential repetition, correcting the treatment
of conditioning on commitment and replacing the earlier
conditional guarantee with an unconditional security bound,
without the finite-size correction to the correctness gap.
The resulting bound provides improved quantitative security parameters
for experimental implementations.
\end{abstract}

\section{Introduction}

Quantum position verification (QPV) combines relativistic signalling constraints with quantum information to certify the spatial location of an untrusted prover \prover{} at a claimed position $pos$. In a one-dimensional setting, two verifiers \Verifzero{} and \Verifone{} send information toward the claimed location and accept only if the prover returns correct answers within the time window consistent with being located at $pos$. Quantum information is essential for this task: even under
computational assumptions, protocols relying solely on classical
information can be broken in this setting by adversaries who copy
and redistribute the transmitted information~\cite{OriginalPositionBasedCryptChandran2009},
whereas the no-cloning theorem prevents such a strategy for
arbitrary quantum states.

Despite this no-cloning advantage, unconditional security of QPV is impossible. Adversaries sharing arbitrarily large amounts of entanglement can simulate an honest prover and break any QPV protocol~\cite{Buhrman_2014,Beigi_2011}. Consequently, much of the literature has focused on restricted—but still powerful—adversarial models, most commonly by bounding the amount of pre-shared entanglement~\cite{PatentKentANdOthers,OriginalQPV_Kent2011,Lau_2011,https://doi.org/10.48550/arxiv.1504.07171,Chakraborty_2015,speelman2016instantaneous,dolev2019constraining,dolev2022non,gonzales2019bounds,cree2022code,bluhm2022single,gao2016quantum,escolàfarràs2024quantumcloninggameapplications,Unruh_2014_QPV_random_oracle,liu2021beating,amer2024certified,girish2026privateproofs,bluhm2026complexitytheorynonlocalquantum,cleve2026lowerboundsnonlocalcomputation,asadi2024lineargateboundsnatural,asadi2025ranklowerboundsnonlocal,escolàfarràs2025quantumpositionverificationshot}. When the number of pre-shared entangled qubits is bounded (bounded entanglement) by $q$, this setting is referred to as the BE$(q)$ model. Moreover, QPV has recently progressed beyond a purely theoretical primitive, with the first experimental demonstrations reported in laboratory settings~\cite{loeffler2025towards,kavuri2026quantumpositionverificationremote,fanyuan2026relativisticpositionverificationcoherent}.

Among the most studied protocols is the BB84-based scheme \QPVBB{}~\cite{PatentKentANdOthers,OriginalQPV_Kent2011}. 
In this protocol, one verifier sends a BB84 state while the other sends the corresponding basis information, timed to arrive simultaneously at the claimed location. 
The prover measures the received qubit in the specified basis and returns the outcome to the verifiers, who sequentially repeat the procedure to make a statistically reliable decision. While \QPVBB{} is secure against adversaries without pre-shared entanglement~\cite{Buhrman_2014,TomamichelMonogamyGame2013}, it admits a perfect attack if adversaries pre-share a single EPR pair. However, a modified version introducing classical information of size $n$, \QPVBBf{}, was shown to remain secure even against attackers sharing a linear amount of entanglement\footnote{See \cite{asadi2024lineargateboundsnatural} for further discussion of the resource requirements.} in $n$. The
\QPVBBf{} protocol
replaces the single basis bit with classical strings \(x,y\in\{0,1\}^n\) and derives the measurement basis as \(z=f(x,y)\)~\cite{Buhrman_2013,bluhm2022single}. 
More recently, the protocol \QPVBBfparallel{}, where all the BB84 states are sent in parallel, was shown to be the first single-execution QPV scheme with negligible accepting probability against adversaries sharing a linear amount (in the size of the classical information) of bounded entanglement~\cite{escolàfarràs2025quantumpositionverificationshot}. Furthermore, similarly as for \QPVBBf, its security does not require quantum information to propagate at the speed of light, making it attractive for implementations in quantum networks. From an implementation perspective, single-execution protocols are particularly attractive. They also offer practical advantages when QPV is employed as a subroutine in larger cryptographic constructions, such as quantum key distribution~\cite{kon2025quantumsecurekeyexchange} and message authentication~\cite{Buhrman_2014}.

However, from the implementation point of view, photon loss remains one of the main obstacles to the practical deployment of QPV. In optical fibers, transmission loss increases exponentially with distance, and over sufficiently long distances only a small fraction of the transmitted qubits are successfully received. Most QPV protocols, including \QPVBB{}, \QPVBBf{}, and \QPVBBfparallel{}, have been analyzed under the idealized assumption that all quantum systems are successfully transmitted to the prover. This assumption is unrealistic in practice, since the following trivial lossy attack applies to all of these protocols: an adversary intercepts and measures each qubit in a randomly chosen basis and only provides an answer when the chosen basis is correct, declaring photon loss otherwise. Consequently, these protocols become fundamentally insecure whenever the transmission rate falls below 50\%.

A variety of approaches have been developed to address this challenge. Early protocols achieved tolerance to a constant fraction of loss, but required quantum communication at the speed of light and became insecure in the presence of even minimal amounts of pre-shared entanglement~\cite{OtherExtentionBB84Qi_Siopsis2015,FlorianThesis}. Later works showed that security against bounded-entanglement adversaries can be maintained without requiring fast quantum communication for constant fractions of loss~\cite{Escol_Farr_s_2023}, making the protocol suitable for relatively short distances but limiting at longer ranges. Subsequently,~\cite{allerstorfer2023makingexistingquantumposition} introduced a different approach: before learning the classical information that determines the measurement basis, the prover must commit to whether the quantum system has arrived. This modification, applied to \QPVBBf{}, yields the protocol \cQPVBBf{}, which retains the security guarantees of the original protocol against bounded-entanglement adversaries while additionally tolerating arbitrarily large amounts of loss.

The sequential argument
of~\cite{allerstorfer2023makingexistingquantumposition},
however, relies on the sequential structure of the protocol
and does not extend to the parallel setting.
In the sequential case, adversaries must repeatedly announce
commitment bits and are detected if their commitments disagree
in any round. The proof uses this repeated risk of rejection
to control the advantage associated with inconsistent commitments.
In a single-execution protocol, by contrast, the verifiers
receive only one commitment string from each attacker, and the
analysis must account for all admissible strings and their
associated conditional states. Consequently, the sequential
argument does not establish security for the parallel commitment
version of \QPVBBfparallel{}.
Single-execution protocols tolerating arbitrary loss have also
been developed~\cite{SWAP_protocol_Rene_et_all}, but they are
insecure in the presence of minimal pre-shared entanglement
and require quantum communication at the speed of light,
limiting their applicability.

In this work, we bridge this gap. We show that the commitment
paradigm introduced
in~\cite{allerstorfer2023makingexistingquantumposition}
can be adapted to \QPVBBfparallel{}, while retaining the
security guarantees of the original protocol. We denote the
resulting commitment-based protocol by \cQPVBBfparallel{}.
This yields a QPV protocol that simultaneously achieves
(i)~negligible soundness error in a single execution,
(ii)~security against adversaries sharing an amount of
entanglement linear in the size of the classical information,
(iii)~no requirement that quantum information propagate at
the speed of light, and
(iv)~tolerance to arbitrarily large transmission losses,
while remaining robust against noise rates up to \(3.7\%\).
We prove that the adversarial acceptance probability decays
exponentially in a threshold parameter \(k\), the minimum
number of committed qubits required for acceptance.

An earlier version of this manuscript contained an error in
the parallel security proof: the treatment of the attackers'
local quantum instruments was incomplete. The present version
corrects this treatment while retaining the two main ingredients
of the earlier proof: the reduction to a
monogamy-of-entanglement game and the counting argument.

We also revisit the sequential commitment protocol
of~\cite{allerstorfer2023makingexistingquantumposition}.
Conditioning on a conclusive commitment can change the
distribution of the classical inputs, a point not accounted
for in the earlier sequential analysis. We correct this step
by keeping the classical inputs together with the quantum
state and applying the gentle measurement lemma, which is the key ingredient used in~\cite{allerstorfer2023makingexistingquantumposition}, to their
joint state. Combining the resulting correctness estimate
with the risk of inconsistent commitments gives an
unconditional acceptance bound for adaptive attacks.
For a fixed total number of rounds \(m\) and a minimum
number \(k\) of conclusive rounds, this bound decays
exponentially in \(k\), without the
\(O(k^{-1/4})\) correction to the correctness gap appearing
in the earlier estimate.

\emph{Use of AI}. Both revisions were developed with assistance from ChatGPT.
For the parallel proof, this assistance helped us resolve
the difficulty in the treatment of quantum instruments while
preserving parts of the previous proof approach. For the sequential
analysis, we combined our ideas with suggestions obtained
through discussions with ChatGPT to develop the corrected
conditioning argument and the resulting security bound.

\subsection{Results}

Let \(h(\cdot)\) denote the binary entropy function with
base-two logarithms. For an error parameter
\(\gamma\in[0,1/2)\), define
\begin{equation}\label{eq:lambda_gamma}
\lambda_\gamma
:=
2^{h(\gamma)}
\left(\frac12+\frac{1}{2\sqrt2}\right).
\end{equation}
In particular, \(\lambda_\gamma<1\) whenever
\(\gamma<0.037\).
This parameter governs the security bound for
the parallel BB84-based protocol
\QPVBBfparallel{}~\cite{escolàfarràs2025quantumpositionverificationshot}: under the corresponding bounded-entanglement assumptions,
the adversarial acceptance probability is bounded, approximately, by~\(\lambda_\gamma^m\).

Our main result establishes security of its commitment version,
\cQPVBBfparallel{}, in the presence of arbitrary transmission
loss. The bound depends on the minimum number \(k\) of
committed positions required for acceptance, rather than
the total number \(m\) of transmitted qubits.
With \(q\) denoting the bound on the number of qubits held by
each attacker before intercepting the classical inputs,
Theorem~\ref{thm:main} shows that, if
\[
4q\leq n-\frac{k}{4}-\log_2 n-7,
\]
then, with probability at least \(1-2^{-n}\) over a uniformly
random function
\({f:\{0,1\}^n\times\{0,1\}^n\to\{0,1\}^m}\),
every {BE}\((q)\) attack satisfies
\[
\Pr[\acc]\leq3\lambda_\gamma^k.
\]

For \(\lambda_\gamma<1\), the adversarial acceptance probability
decreases exponentially in \(k\). As in the underlying protocol, quantum communication need not occur at the speed of light.

We also revisit sequential repetition of the commitment
protocol introduced
in~\cite{allerstorfer2023makingexistingquantumposition}.
We correct the treatment of the classical input distribution
after conditioning on a conclusive commitment. Denoting by \(p^*\) an upper bound on the acceptance probability
of the original one-round protocol~\QPVBBf, 
Theorem~\ref{thm:seq-soundness} shows that, for a fixed number
\(m\geq k\) of sequential rounds and \(0\leq\gamma<1-p^*\),
every adaptive attack to~\cQPVBBf{} satisfies
\[
\Pr[\acc]
\leq
\exp\left[-\frac{2k(1-\gamma-p^*)^2}{19}\right].
\]
Compared with the earlier sequential estimate, this bound
accounts for rejection due to inconsistent commitments directly
in the acceptance probability and removes the finite-size
correction $O(k^{-1/4})$ to the correctness gap \(1-\gamma-p^*\).

\section{Preliminaries}
All logarithms are base 2 unless written as \(\ln\). Let $\mathcal{H}$, $\mathcal{H'}$ be finite-dimensional Hilbert spaces. We denote by $\mathcal{B}(\mathcal{H},\mathcal{H'})$ the set of bounded operators from $\mathcal{H}$ to $\mathcal{H'}$ and $\mathcal{B}(\mathcal{H})=\mathcal{B}(\mathcal{H},\mathcal{H})$.  Denote by $\mathcal{S}(\mathcal{H})$ the set of quantum states on $\mathcal{H}$,~i.e.\ $\mathcal{S}(\mathcal{H})=\{\rho\in\mathcal{B}(\mathcal{H})\mid \rho\geq0, \tr{\rho}=1\}$. A linear map $\mathcal{E}:\mathcal{B}(\mathcal H)\rightarrow \mathcal{B}(\mathcal H')$ is a \emph{quantum channel} if it is completely positive and trace preserving (CPTP). 

\begin{lemma} \label{lemma thm Kraus decomposition} \emph{(Kraus representation \cite{Kraus1971311})}. A linear map $\Phi$ is completely positive and trace non-increasing if and only if there exist bounded operators $\{K_i\}_{i=1}^r$ such that for all density operators~$\rho$,
\begin{equation}
    \Phi(\rho)=\sum_{i=1}^rK_i\rho K_i^{\dagger},
\end{equation}
with $\sum_{i=1}^rK_i^{\dagger} K_i\leq \mathbb{I}$, where $r$ is the Kraus rank. Moreover, $\Phi$ is trace-preserving,~i.e.\ a quantum channel, if and only if $\sum_{i=1}^rK_i^{\dagger} K_i = \mathbb{I}$.
\end{lemma}
Let $\Omega$ be a finite outcome set. A \emph{quantum instrument} $\mathcal{I}$ is a set of completely positive linear maps $\{\mathcal{I}_i\}_{i\in\Omega}$ such that $\sum_{i\in\Omega}\mathcal{I}_i$ is trace preserving. Given the quantum state $\mathcal{\rho}\in\mathcal{S}(\mathcal{H})$, the probability of obtaining outcome $i$ is given by $\tr{\mathcal{I}_i(\rho)}$ and the sub-normalized output state upon outcome $i$ is $\mathcal{I}_i(\rho)$.

\begin{lemma}\label{lemma: instrument and channel} \emph{(Thm 7.2 in \cite{hayashi2016quantum})}
Let $\mathcal{I}=\{\mathcal{I}_i\}_{i\in\Omega}$ be an instrument, and $\{M_i\}_i$ its corresponding POVM, i.e.\ $\mathcal{I}_i^\dagger(\mathbb{I}) = M_i$. 
Then, for every $i\in\Omega$, there exists a quantum channel (CPTP map) $\mathcal{E}_i$ such that  
    \begin{equation}
        \mathcal{I}_i(\rho)=\mathcal{E}_i \left(\sqrt{M_i}\rho\sqrt{M_i}\right).
    \end{equation}
\end{lemma}

Throughout the paper, \(\|\cdot\|_1\) denotes the trace norm.

\begin{lemma}[Gentle measurement lemma {\cite{winter1999coding}}]
\label{lem:seq-gentle}
Let \(\rho\) be a density operator and \(0\leq M\leq I\).
If \(\Tr(M\rho)\geq1-\varepsilon\), with
\(0\leq\varepsilon<1\), then the normalized post-measurement state
\[
\rho_M=\frac{\sqrt M\,\rho\,\sqrt M}{\tr{M\rho}}
\]
satisfies
\(
\|\rho_M-\rho\|_1\leq2\sqrt{\varepsilon}.
\)
\end{lemma}

\begin{lemma}[{\cite{wilde2017quantum}}]
\label{lem:measurement-stability}
Let \(\rho,\sigma\) be density operators on the same registers,
let \(\mathcal E\) be a quantum channel, and let \(\{M,I-M\}\)
be a binary POVM on its output. Then
\[
\left|
\Tr[M\mathcal E(\rho)]
-\Tr[M\mathcal E(\sigma)]
\right|
\leq \frac12\|\rho-\sigma\|_1.
\]
\end{lemma}

\begin{lemma}[Hoeffding's lemma for a binary variable
{\cite{hoeffding1963probability}}]\label{lem:seq-hoeffding}
Let \(Y\in\{0,1\}\) be a random variable with mean
\(\mu=\mathbb E[Y]\). Then, for every \(\theta\in\mathbb R\),
\[
\mathbb E[\exp(\theta(Y-\mu))]
\leq\exp(\theta^2/8).
\]
\end{lemma}

\section{Commitment QPV Structure for Multiple Qubits}

In this section, we present the BB84-based protocol introduced in~\cite{escolàfarràs2025quantumpositionverificationshot}, denoted by \QPVBBfparallel{}. We then introduce the corresponding commitment modification, originally proposed in~\cite{allerstorfer2023makingexistingquantumposition} for the case $m=1$.

In the \QPVBBfparallel{} protocol, one verifier prepares $m$ qubits, each independently chosen in one of the four BB84 states---either in the computational basis ($\ket{0}, \ket{1}$) or in the Hadamard basis ($\ket{+}, \ket{-}$)---and sends them to the prover together with an $n$-bit string $x$. The other verifier sends an $n$-bit string $y$. Upon reception of the classical inputs, the prover computes $z = f(x,y)$ and measures the $i$-th qubit in the basis specified by $z_i$, obtaining an outcome string $w \in \{0,1\}^m$, which is broadcast to both verifiers. The full specification is given in \cref{fig:qpv-bb-parallel} and schematically represented in \cref{fig:parallel_BB84}.

In~\cite{escolàfarràs2025quantumpositionverificationshot}, it was shown that for a suitable class of functions $f$, the protocol is secure against adversaries in the BE$(O(n))$ model, i.e., adversaries pre-sharing a number of entangled qubits linear in the size of the classical information.

\begin{figure}[h]
\hrule
\vspace{0.5em}
\textbf{The \QPVBBfparallel{} protocol~\cite{escolàfarràs2025quantumpositionverificationshot}.} 
\vspace{0.5em}
\hrule
\vspace{0.5em}
Let $n>m$, where $n$ denotes the size of the classical inputs and $m$ the number of qubits used in the protocol. Let $f:\{0,1\}^n \times \{0,1\}^n \to \{0,1\}^m$ be a publicly known function and $\gamma \in \left[0,\frac{1}{2}\right)$ be an error parameter.

\begin{enumerate}
    \item \textbf{Preparation.}  The verifiers \Verifzero{} and \Verifone{} secretly sample uniform random strings
    \(
       {x,y \in \{0,1\}^n}, v \in \{0,1\}^m     \) and compute $z := f(x,y) \in \{0,1\}^m$. \Verifzero{} prepares   
        $H^{z}\ket{v} = \otimes_{i=1}^m H^{z_i}\ket{v_i}$.

    \item \textbf{Distribution.} \Verifzero{} sends $H^{z}\ket{v}$ and $x$ to \prover, while \Verifone{} sends $y$. The classical strings $x$ and $y$ are constrained to propagate at the speed of light and are timed to arrive simultaneously at ${pos}$, whereas the quantum states may be transmitted arbitrarily slowly but must be available at ${pos}$ upon arrival of the classical messages.

    \item \textbf{Measurement.} The prover computes $z=f(x,y)$ and measures each qubit in the basis defined by $z_i$, obtaining $w \in \{0,1\}^m$, which is broadcast to both verifiers.

    \item \textbf{Verification.} The verifiers accept if both of them receive the same $w$ within the correct timing window and 
    \begin{equation}\label{eq:accept w}
          d_H(v,w) \leq \gamma m \,\,\,\,\,\,\,\,\,\textrm{ (accept $w$)}.
    \end{equation}
\end{enumerate}

\vspace{0.5em}
\hrule
\caption{Description of the \QPVBBfparallel{} protocol.}
\label{fig:qpv-bb-parallel}
\end{figure}

For suitable functions \(f\), the security analysis
of~\cite{escolàfarràs2025quantumpositionverificationshot} bounds
the acceptance probability of bounded-entanglement attacks by
an expression whose leading dependence on the number of qubits
is \[\lambda_\gamma^m,\] up to correction terms.
Thus, in the parameter regime where \(\lambda_\gamma<1\),
parallel repetition gives an exponentially decreasing acceptance
probability. We refer to that work for the precise parameter
conditions and finite-size bound.

\begin{figure}[h]
    \centering
    \begin{tikzpicture}[node distance=3cm, auto]
    \node (A) {\Verifzero{}};
    %\node [left=1cm of A] {};
    \node [right=of A] (B) {\prover{}};
    \node [right=of B] (C) {\Verifone};
    %\node [right=1cm of C] {};
    \node [below=of A] (D) {};
    \node [below=of B] (E) {};

    \node [above=-0.1cm of A] (N) {};
    \node [right=0.8cm of N] {$\diagdots[120]{2.5em}{0.1em}$};
    
    \node [right=0.8cm of A] (M) {};
    \node [left=1.5cm of C] (M2) {};
    
    \node [below=of C] (F) {};
    \node [below=of D] (G) {};%\Verifzero{}
    \node [below=of E] (H) {};
    \node [below=of F] (I) {};%\Verifone{}
    \node [left= 4.5cm of E] (J) {};
    \node [below= 3cm of J] (K) {};
    \node [above= 3cm of J] (L) {};
    \node [right=0.1 of E](P fxy){$f(x,y)=z_1\ldots z_m$};

    \draw [->, transform canvas={xshift=0pt, yshift=0pt}, quantum] (M) -- (E) node[midway] (x) {} ;
    \draw [->, transform canvas={xshift=2pt, yshift=0pt}, quantum] (M) -- (E) node[midway] {} ;
    \draw [->, transform canvas={xshift=4pt, yshift=0pt}, quantum] (M) -- (E) node[midway] {} ;
    \draw [->, transform canvas={xshift=-2pt, yshift=0pt}, quantum] (M) -- (E) node[midway] {} ;
    %\draw [->, transform canvas={xshift=-4pt, yshift=0pt}, quantum] (M) -- (E) node[midway] {} ;
    \node [right=0.8cm of N,transform canvas={xshift=-4pt, yshift=0pt}] {$\diagdots[120]{2.5em}{0.1em}$};
    \node [right=0.8cm of N,transform canvas={xshift=-2pt, yshift=0pt}] {$\diagdots[120]{2.5em}{0.1em}$};
    \node [right=0.8cm of N,transform canvas={xshift=4pt, yshift=0pt}] {$\diagdots[120]{2.5em}{0.1em}$};
    \node [right=0.8cm of N,transform canvas={xshift=2pt, yshift=0pt}] {$\diagdots[120]{2.5em}{0.1em}$};

    \draw [->] (A) -- (E);
    \draw [->] (C) -- (E);
    \draw [->] (E) -- (I) node[midway] (q) {}; 
    %\draw [->,transform canvas={xshift=-4pt, yshift=0pt}] (E) -- (I) ;
    \draw [->,transform canvas={xshift=4pt, yshift=0pt}] (E) -- (I) ;
    \draw [->,transform canvas={xshift=-2pt, yshift=0pt}] (E) -- (I) ;
    \draw [->,transform canvas={xshift=2pt, yshift=0pt}] (E) -- (I) ;

    \draw [][->] (E) -- (G); %line width=0.4mm, dotted
    %\draw [->,transform canvas={xshift=-4pt, yshift=0pt}] (E) -- (G) ;
    \draw [->,transform canvas={xshift=4pt, yshift=0pt}] (E) -- (G) ;
    \draw [->,transform canvas={xshift=-2pt, yshift=0pt}] (E) -- (G) ;
    \draw [->,transform canvas={xshift=2pt, yshift=0pt}] (E) -- (G) ;

    \draw [->] (L) -- (K) node[midway] {time};

    \node[below=0.5cm of B, xshift=-20pt, fill opacity=0.7,      text opacity=1] {$\otimes_{i=1}^m\ket{\phi_i}$};
    
    \node[left=1.4cm of x, transform canvas={xshift=+ 2pt, yshift = +2 pt}] {$x\in\{0,1\}^n$};
    \node[right = 2.9cm of x, transform canvas={xshift=+ 2pt, yshift = +2 pt}] {$y \in \{0,1\}^n$};
    \node[left = 2.5cm of q, fill=white, fill opacity=0.7, text opacity=1, xshift=-14pt, yshift=-15pt] (V0ans) {$w\in\{0,1\}^m$};%\in\{0,1\}^m
    \node[right = 1cm of V0ans, fill=white, fill opacity=0.7, text opacity=1, xshift=45pt] (V0ans) {$w\in\{0,1\}^m$};%\in\{0,1\}^m

    \node [above=0.5cm of A] (posV00) {};
    \node [left=1cm of posV00] (posV0) {};
    \node [above=0.5cm of C] (posV11) {};
    \node [right=1cm of posV11, xshift=-2em] (posV1) {};
    \draw [->] (posV0) -- (posV1) node[midway,yshift=3pt,xshift=11pt] {position};

    \node[above=0.5cm of A, yshift=-6pt, xshift=-2pt](V0pos){$|$};
    \node[right= of V0pos, xshift=10pt](Ppos){$|$};
    \node[below=0.2cm of Ppos, yshift=10pt](pos){$pos$};
    \node[right=7cm of V0pos, xshift=-3pt](V1pos){$|$};

    \end{tikzpicture}
\caption{Schematic representation of \QPVBBfparallel. }

\label{fig:parallel_BB84}
\end{figure}

The above result establishes security guarantees for
\QPVBBfparallel{}. However, the protocol assumes that all
\(m\) transmitted qubits are available to the prover when the
classical information arrives, unless missing qubits are counted
as errors within the tolerance \(\gamma\). In realistic optical implementations, transmission losses are unavoidable and may be substantial over long distances. As discussed in the introduction, this creates a fundamental vulnerability: an adversary can measure each incoming qubit in a randomly chosen basis and only report outcomes for those qubits for which the guess was correct, declaring loss otherwise. Since the adversary succeeds with probability 1/2 on each qubit, any protocol that does not explicitly constrain loss events becomes insecure once the transmission rate drops below 50\%.

To address this issue, we adopt the \emph{commitment} modification introduced in~\cite{allerstorfer2023makingexistingquantumposition}. The central idea is that the prover must declare, before learning the basis information, which qubits have been successfully received. This prevents adversaries from using knowledge of the basis string $z=f(x,y)$ to selectively decide which qubits are reported as lost. We apply this idea to the parallel protocol \QPVBBfparallel{}, obtaining the commitment protocol \cQPVBBfparallel{} described in \cref{fig:description_c-parallel_BB84}. The protocol is identical to \QPVBBfparallel{} except for the addition of a commitment phase prior to the measurement step. Specifically, after receiving the quantum systems but before learning the classical inputs $x$ and $y$, the prover broadcasts a commitment string $c\in\{0,1\}^m$ indicating which qubits have been successfully received, where $c_i=1$ if the $i$-th qubit has been received and $c_i=0$ otherwise. This corresponds to Step~3 of \cref{fig:description_c-parallel_BB84} and constitutes the essential modification relative to \QPVBBfparallel{}. The verifiers proceed only if the commitment contains at least $k$ declared arrivals, for a certain pre-defined threshold $k$. Once the classical information becomes available, the prover measures only the committed qubits and the verification test is restricted to this subset. Apart from this restriction, the measurement and verification procedures are identical to those of the original \QPVBBfparallel{} protocol, applied to the subset $C=\{i\mid c_i=1\}$ instead of the full set of $m$ qubits.

The parameter $k$ aggregates all sources of qubit loss in an implementation, including transmission loss, coupling inefficiencies, detector inefficiency, and failures of photon presence detection, which can all be captured by a single loss model. The experimental feasibility of this commitment paradigm is discussed in~\cite{allerstorfer2023makingexistingquantumposition}.

In the next section, we carry out the security analysis of
\cQPVBBfparallel{} and show that the commitment modification
provides a soundness guarantee that accommodates arbitrary loss.
Under the stated bounded-entanglement and parameter assumptions,
the adversarial acceptance probability is at most
\(3\lambda_\gamma^k\). For \(\lambda_\gamma<1\), this bound
decays exponentially in \(k\), similarly to the dependence on
\(m\) in the original protocol discussed above. Here \(k\),
the minimum number of committed positions required for acceptance,
replaces \(m\) as the parameter governing the security bound.

\begin{figure}[h!]
\hrule
\vspace{0.5em}
\textbf{The \cQPVBBfparallel{} protocol.}
\vspace{0.5em}
\hrule
\vspace{0.5em}

Let $n>m$, where $n$ denotes the size of the classical inputs and $m$ the number of qubits used in the protocol. Let $f:\{0,1\}^n \times \{0,1\}^n \to \{0,1\}^m$ be a publicly known function, let $\gamma \in \left[0,\frac{1}{2}\right)$ be an error parameter, and let $k \in [m]$ be a threshold parameter.

\begin{enumerate}
    \item \textbf{Preparation.} The verifiers \Verifzero{} and \Verifone{} secretly sample uniform random strings
    \(
       {x,y \in \{0,1\}^n}, v \in \{0,1\}^m     \) and compute $z := f(x,y) \in \{0,1\}^m$. \Verifzero{} prepares   
        $H^{z}\ket{v} = \otimes_{i=1}^m H^{z_i}\ket{v_i}$.

    \item \textbf{Distribution.} \Verifzero{} sends $H^{z}\ket{v}$ to \prover{}, while \Verifzero{} and \Verifone{} send $x$ and $y$, respectively. The classical strings are constrained to propagate at the speed of light and are timed to arrive simultaneously at ${pos}$. 
    The quantum qubits are sent earlier, with no additional constraints beyond being available at $\mathsf{pos}$ \emph{prior} to the arrival of the classical inputs.
    
    \item \textbf{Commitment.} Upon receiving the quantum systems, \prover{} determines which qubits have arrived and broadcasts a commitment string
\(
    c \in \{0,1\}^m,
\)
where $c_i = 1$ if the $i$-th qubit has been received and $c_i = 0$ otherwise. The commitment must reach the verifiers in time to have been
sent from the claimed position before the classical inputs arrive there. The verifiers accept the commitment if
\begin{equation}
    |c| \geq k \,\,\,\,\,\,\,\,\,\textrm{ (accept $c$)},
\end{equation}
   otherwise, the protocol aborts (and the verifiers reject).

    \item \textbf{Measurement.} Upon receiving $x$ and $y$, \prover{} computes $z = f(x,y)$. Let
\(
    C := \{i \in [m] : c_i = 1\}
\)
be the set of indices corresponding to received qubits. The prover measures the qubits indexed by $C$ in the bases specified by $z_C$, obtaining outcomes 
\(
    w_C \in \{0,1\}^{|c|},
\)
which are broadcast to both verifiers.

    \item \textbf{Verification.} The verifiers accept if: both receive the same strings $c$ and $w_C$, all messages arrive within the timing constraints consistent with ${pos}$, and 
        \begin{equation}
             d_H(v_C,w_C) \leq \gamma |c| \,\,\,\,\,\,\,\,\,\textrm{ (accept $w_C$)}.
\end{equation}
\end{enumerate}

\vspace{0.5em}
\hrule
\caption{Description of the \cQPVBBfparallel{} protocol}
\label{fig:description_c-parallel_BB84}
\end{figure}

\begin{figure}[h]
    \centering
    \begin{tikzpicture}[node distance=3cm, auto]
    \node (A) {\Verifzero{}};
     
    %\node [left=1cm of A] {};
    \node [right=of A] (B) {\prover{}};
    \node [right=of B] (C) {\Verifone};
    %\node [right=1cm of C] {};
    \node [below=of A] (D) {};
    \node [below=of B] (E) {};
    %\draw [-,dotted] (A) -- (E);

    %\node [above=-0.1cm of A] (N) {};
    %\node [right=0.8cm of N] {$\diagdots[120]{2.5em}{0.1em}$};
    
    \node [right=1.6cm of A] (M) {};
    \node [left=1.5cm of C] (M2) {};
    
    \node [below=of C] (F) {};
    \node [below=of D] (G) {};%\Verifzero{}
    \node [below=of E] (H) {};
    \node [below=of F] (I) {};%\Verifone{}
    \node [left= 4.5cm of E] (J) {};
    \node [below= 4.4cm of J] (K) {};
    \node [above= 3cm of J] (L) {};
    
\path (M) -- (E) coordinate[pos=0.1] (mid00);
\path (M) -- (E) coordinate[pos=0.3] (mid0);
\path (M) -- (E) coordinate[pos=0.25] (mid1);
\path (M) -- (E) coordinate[pos=0.4] (mid2);
\path (M) -- (E) coordinate[pos=0.55] (mid3);
\path (M) -- (E) coordinate[pos=0.7] (mid4);

% full channels

\draw [->, transform canvas={xshift=-3pt}, quantum]
    (M) -- (E) node[midway] (x) {};
\draw [->, transform canvas={xshift=0pt, yshift=2}, quantum]     (M) -- (E);
\draw [->, transform canvas={xshift=3pt, yshift=4}, quantum]     (M) -- (mid2);
\draw [->, transform canvas={xshift=-6pt, yshift=-2}, quantum]     (M) -- (mid3);
 \draw [->, transform canvas={xshift=-9pt, yshift=-4}, quantum]     (M) -- (mid1);
\draw [->, transform canvas={xshift=6pt, yshift=6}, quantum]     (M) -- (mid0);

  %Extra ones
  %\draw [-, transform canvas={xshift=9pt, yshift=6}, quantum]     (M) -- (mid3);
     %\draw [-, transform canvas={xshift=12pt, yshift=8}, quantum]     (M) -- (mid2);
     %\draw [-, transform canvas={xshift=15pt, yshift=10}, quantum]     (M) -- (mid2);
     
 % \draw [->, transform canvas={xshift=-12pt, yshift=-8}, quantum]     (M) -- (mid1);
  %\draw [-, transform canvas={xshift=-15pt, yshift=-10}, quantum]     (M) -- (mid1);

    %\node [right=0.8cm of N,transform canvas={xshift=-4pt, yshift=0pt}] {$\diagdots[120]{2.5em}{0.1em}$};
    %\node [right=0.8cm of N,transform canvas={xshift=-2pt, yshift=0pt}] {$\diagdots[120]{2.5em}{0.1em}$};
    %\node [right=0.8cm of N,transform canvas={xshift=0pt, yshift=0pt}] {$\diagdots[120]{2.5em}{0.1em}$};
    %\node [right=0.8cm of N,transform canvas={xshift=4pt, yshift=0pt}] {$\diagdots[120]{2.5em}{0.1em}$};
    %\node [right=0.8cm of N,transform canvas={xshift=2pt, yshift=0pt}] {$\diagdots[120]{2.5em}{0.1em}$};
    %\node [right=0.8cm of N,transform canvas={xshift=6pt, yshift=0pt}] {$\diagdots[120]{2.5em}{0.1em}$};

    \node [below=0.90cm of A] (V0) {};
    \node [below=0.90cm of C] (V1) {};
    \node [below=1cm of E] (P) {};
    \node [right=0.1 of P, xshift=-5pt](P fxy){$f(x,y)$};
    \node [below=1cm of G] (V0_ans) {};
    \node [below=1cm of I] (V1_ans) {};

    \draw [-, transform canvas={xshift=-2pt, yshift=0pt}, quantum] (E) -- (P) node[midway] {} ;
    \draw [-, transform canvas={xshift=2pt, yshift=0pt}, quantum] (E) -- (P) node[midway] {} ;

    \draw [->] (V0) -- (P);
    \draw [->] (V1) -- (P);
    \draw [->] (E) -- (I) node[midway] (q) {}; 
    \draw [->,transform canvas={xshift=-4pt, yshift=0pt}] (E) -- (I) ;
    \draw [->,transform canvas={xshift=4pt, yshift=0pt}] (E) -- (I) ;
    \draw [->,transform canvas={xshift=-2pt, yshift=0pt}] (E) -- (I) ;
    \draw [->,transform canvas={xshift=2pt, yshift=0pt}] (E) -- (I) ;
    \draw [->,transform canvas={xshift=6pt, yshift=0pt}] (E) -- (I) ;

    \draw [][->] (E) -- (G); %line width=0.4mm, dotted
    \draw [->,transform canvas={xshift=-4pt, yshift=0pt}] (E) -- (G) ;
    \draw [->,transform canvas={xshift=4pt, yshift=0pt}] (E) -- (G) ;
    \draw [->,transform canvas={xshift=-2pt, yshift=0pt}] (E) -- (G) ;
    \draw [->,transform canvas={xshift=2pt, yshift=0pt}] (E) -- (G) ;
    \draw [->,transform canvas={xshift=6pt, yshift=0pt}] (E) -- (G) ;

     \draw [->, transform canvas={xshift=-2pt, yshift=0pt}] (P) -- (V0_ans) node[midway] {} ;
    \draw [->, transform canvas={xshift=2pt, yshift=0pt}] (P) -- (V0_ans) node[midway] {} ;

    \draw [->, transform canvas={xshift=-2pt, yshift=0pt}] (P) -- (V1_ans) node[midway] {} ;
    \draw [->, transform canvas={xshift=2pt, yshift=0pt}] (P) -- (V1_ans) node[midway] {} ;

    \draw [->] (L) -- (K) node[midway] {time};

    \node[below=0.2cm of B, xshift=-38pt, yshift=20pt, fill=white, fill opacity=1,      text opacity=1] {$\otimes_{i=1}^m\ket{\phi_i}$};%{$\otimes_{i=1}^mH^{f(x,y)_i}\ket{v_i}$};
    
    \node[below=1cm of V0, xshift=30pt] {$x$};
    \node[below=1cm of V1, xshift=-30pt] {$y$};

    \node[left = 2.5cm of q, fill opacity=0.7, text opacity=1, xshift=-45pt, yshift=-10pt] (V0ans) {$c$};%\in\{0,1\}^m
    
    \node[right = 4.6cm of V0ans,  fill opacity=0.7, text opacity=1] (V1ans) {$c$};%\in\{0,1\}^m

    \node[below=1.4cm of V0ans, fill=white, xshift=0pt, fill opacity=0.7, text opacity=1] (V0ans2) {$w_{C}\in\{0,1\}^{\abs{C}}$};
    \node[below=1.4cm of V1ans, xshift=10pt, fill=white, fill opacity=0.7, text opacity=1, ] (V0ans2) {$w_{C}\in\{0,1\}^{\abs{C}}$};

    \node [above=0.5cm of A] (posV00) {};
    \node [left=1cm of posV00] (posV0) {};
    \node [above=0.5cm of C] (posV11) {};
    \node [right=1cm of posV11, xshift=-2em] (posV1) {};
    \draw [->] (posV0) -- (posV1) node[midway,yshift=3pt,xshift=11pt] {position};

    \node[above=0.5cm of A, yshift=-6pt, xshift=-2pt](V0pos){$|$};
    \node[right= of V0pos, xshift=10pt](Ppos){$|$};
    \node[below=0.2cm of Ppos, yshift=10pt](pos){$pos$};
    \node[right=7cm of V0pos, xshift=-3pt](V1pos){$|$};

    \end{tikzpicture}
\caption{Schematic representation of \cQPVBBfparallel.  Unlike~\cref{fig:parallel_BB84}, 
which depicts the original protocol where all $m$ qubits are measured and the full 
string $w \in \{0,1\}^m$ is returned, here only the committed subset $C$ is measured, 
and the prover returns $w_C \in \{0,1\}^{|c|}$. }

\label{fig:parallel-BB84}
\end{figure}

\subsection{Limitations of the Sequential Repetition Argument}

A natural first attempt to prove security of the parallel commitment
protocol is to extend the techniques used for sequential repetition
in~\cite{allerstorfer2023makingexistingquantumposition}.
To explain where this approach runs into difficulties, we first
briefly describe the sequential argument. We use the corrected
analysis developed in Section~\ref{sec:sequential-revisited},
which accounts for the change in the input distribution after
conditioning on a conclusive commitment and applies the gentle
measurement lemma in a different way than in~\cite{allerstorfer2023makingexistingquantumposition}.

The main idea is a tradeoff between commitment consistency and
answer correctness. For each round, define
\[
\varepsilon
=
\max\bigl\{
\Pr[c_A=0\mid c_B=1],
\Pr[c_B=0\mid c_A=1]
\bigr\},
\]
whenever the round has positive probability of being conclusive.
The gentle measurement lemma then bounds the conditional
correctness probability by \(p^*+3\sqrt{\varepsilon}\), where
\(p^*\) is the corresponding security bound for the original
protocol.

The additive term is useful only when \(\varepsilon\) is small.
However, a large \(\varepsilon\) also means a substantial risk
of inconsistent commitments. In sequential repetition, an attack therefore
cannot repeatedly take advantage of large values of
\(\varepsilon\) while accumulating many conclusive rounds
without risking rejection. The proof combines these two
possibilities over the execution; it does not assume that
\(\varepsilon\) is small in every round.

This is the first difficulty in extending the argument to the
parallel protocol. The attackers announce their commitment strings
in a single execution, so there are no successive commitment checks
over which to accumulate the same rejection risk. The conditional
inconsistency parameters for particular commitments may be close
to one, making the gentle measurement bound trivial. Even a small
additive error that does not decrease with \(k\) would prevent
this argument from giving a desired exponentially decreasing
bound. Nor can we simply regard the coordinates as sequential
rounds: the attackers may act jointly on all qubits, and every
coordinate uses the same classical input pair \(x,y\).

The second difficulty is that there are many possible conclusive
commitments. In a sequential round, only the common commitment
\(c_A=c_B=1\) requires a correctness analysis. In the parallel
protocol, every common string \(c\in\{0,1\}^m\) with
\(|c|\geq k\) is admissible. There are
\[
\sum_{r=k}^{m}\binom{m}{r}
\]
such strings, which can be exponentially many in \(m\).
Each can select a different subset of qubits and give a different
conditional state. Applying the gentle measurement lemma to each state gives an
additive error for each commitment, but there is no guarantee
that these errors remain small enough after averaging over
the commitments.

We therefore use a different proof for the parallel protocol.

\section{Security Analysis of \texorpdfstring{\cQPVBBfparallel}{c-QPVfm}}\label{section:security_parallel}

As is customary in the literature, for the security analysis, we consider the purified version of \cQPVBBfparallel{}, which is equivalent to the original protocol. In this formulation, \Verifzero{} prepares $m$ EPR pairs
\(
\ket{\Phi^+}_{V_1Q_1}\otimes \cdots \otimes \ket{\Phi^+}_{V_mQ_m},
\)
and sends the qubit registers $Q_1,\ldots,Q_m$ to the prover. At a later stage, \Verifzero{} measures the local registers $V_1\ldots V_m=:V$ using the measurement
\[
\left\{
M_v^{f(x,y)}
\right\}_{v\in\{0,1\}^m},
\]
where
\[
M_v^{f(x,y)} := H^{f(x,y)}\ketbra{v}{v}_V H^{f(x,y)}=\otimes_{i=1}^mH^{f(x,y)_i}\ketbra{v}{v}_{V_i}H^{f(x,y)_i}.
\]
where each qubit is measured either in the computational or Hadamard basis according to $f(x,y)$. Since the protocol only verifies the coordinates in $C \subseteq [m]$, the measurement can equivalently be restricted to the subsystem $V_C$. More precisely, it suffices to consider the induced POVM on $V_C$ given by
\[
\left\{
 M_{v_C}^{f(x,y)_C}
\right\}_{v_C\in\{0,1\}^{|C|}},
\]
where $ M_{v_C}^{f(x,y)_C}:=H^{f(x,y)_C}\ketbra{v_C}{v_C} H^{f(x,y)_C}$. Equivalently, only the outcomes on $C$ are relevant for verification, and the remaining registers can be disregarded.

\noindent A general attack on \cQPVBBfparallel, see \cref{fig:attack-parallel_repBB84}, follows the same structure as in \cite{allerstorfer2023makingexistingquantumposition}:
\begin{enumerate} \item Alice intercepts the $m$-qubit state $Q_1\ldots Q_m$ and applies an arbitrary quantum operation to it and to a local register that she possesses, possibly entangling them. She keeps part of the resulting state, and sends the rest to Bob. Since the qubits $Q_1\ldots Q_m$  can be sent arbitrarily slowly by $V_0$ (the verifiers only time the classical information), this happens before Alice and Bob can intercept $x$ and $y$. Denote by $\rho$ their joint state at this stage. 

\item Alice and Bob intercept $x$ and $y$, make a copy and send it to the other attacker, respectively. Because of the relativistic constraints, the attackers have to commit before they receive the classical information from the other party, thus their commitment can only depend on $x$ or $y$, respectively. The most general operation they can perform is to use local quantum instruments $\{\mathcal{I}^x_{c_A}\}_{c_A\in\{0,1\}^m}$ and $\{\mathcal{J}^y_{c_B}\}_{c_B\in\{0,1\}^m}$ on their registers of $\rho$ to determine the commitments $c_A$ and $c_B$, respectively. They send the commitments $c_A, c_B$ to the verifier closest to them at the appropriate time. Assume $c_A=c_B=:c$ and $\abs c\geq k$, otherwise the location is \emph{rejected}. Let 
\begin{equation}\label{eq:post-commitment}
    \widetilde{\mathcal{I}}^{xy}_c(\rho):=\frac{(\mathbb I_V\otimes \mathcal{I}^{x}_{c}\otimes \mathcal{J}^{y}_{c})(\rho)}{\tr{(\mathbb I_V\otimes\mathcal{I}^{x}_{c}\otimes \mathcal{J}^{y}_{c})(\rho)}},
\end{equation}

which is the state conditioned on the commitment $c$, given inputs $x$ and $y$. Denote by  $A:=A_{\text{k},x}A_{\text{s},x}$ and $B=:B_{\text{k},y}B_{\text{s},y}$ the registers of $\widetilde{\mathcal{I}}^{xy}_c(\rho)$ that Alice and Bob hold, where the subscripts k and s denote the registers that will be kept and sent, respectively.

\item Alice sends register $A_{\text{s},x}$ and $x$ to Bob (and keeps register $A_{\text{k},x}$), and Bob sends register $B_{\text{s},y}$ and $y$ to Alice (and keeps register $B_{\text{k},y}$). 

\item Alice and Bob perform POVMs  $\{ P^{xyc}_{a_C}\}_{a_C\in\{0,1\}^{\abs{C}}}$ and $\{ Q^{xyc}_{b_C}\}_{b_C\in\{0,1\}^{\abs{C}}}$ on their local registers $A_{\text{k},x}B_{\text{s},y}=:A_{xy}'$ and $B_{\text{k},y}A_{\text{s},x}=:B_{xy}'$, and answer their outcomes $a_C$ and $b_C$ to their closest verifier, respectively. 
\end{enumerate}

Note that, unlike the previous analyses of the original sequential protocol,
its commitment variant, and the original parallel protocol, we
explicitly allow the division into kept and sent registers to
depend on the local classical inputs \(x\) and \(y\). Our security
analysis therefore also covers this additional freedom in the
attackers' strategies.

An attack on \QPVBBfparallel{} is described by replacing the
commitment instruments
\(\{\mathcal I^x_{c_A}\}_{c_A}\) and
\(\{\mathcal J^y_{c_B}\}_{c_B}\)
with quantum channels \(\mathcal E_A^x\) and
\(\mathcal E_B^y\), respectively, and testing all \(m\) positions.

Throughout our analysis, the {BE}\((q)\) model imposes
a bound of \(q\) qubits on each attacker's quantum registers
both before and after the input-dependent local operation.
Specifically, each attacker initially holds at most \(q\)
qubits of \(\rho\), and, for every local outcome, their
combined kept and sent quantum registers contain at most
\(q\) qubits before communication. 
Classical records and messages are unrestricted.
Smaller quantum input and output spaces can be embedded in
\(\mathbb C^{2^q}\), so we may take both to have dimension
\(d=2^q\) without loss of generality.

\begin{figure}[h]
    \centering
\begin{tikzpicture}[node distance=3cm, auto]
    \node (V0) {\Verifzero};
    \node [above=0.1cm of V0, xshift=-2pt](V0pos){$|$};
    \node [right=1cm of V0pos, xshift=13pt](Apos){$|$};
    \node [right=2.5cm of V0pos, xshift=18pt](Ppos){$|$};
    
    \node [right=4.5cm of V0pos, xshift=10pt](Bpos){$|$};
    \node [right=6.5cm of V0pos, xshift=2pt](V1pos){$|$};
    %\node [left=1cm of V0] {};
    \node [right=of V0] (P) {};

    \node[below=0.2cm of P, yshift=-1pt](psi_0){$\rho$};
    \node[below=0.7cm of P, yshift=-2pt](psi_xy){{$\widetilde{\mathcal{I}}^{xy}_c(\rho)$}};
\node[below=4cm of P, yshift=-4pt](psi_xy){$\widetilde{\mathcal{I}}^{xy}_c(\rho)$};

    \node [right=of P] (V1) {};
    \node [right=of P] (C) {\Verifone};
    \node [right=1.5 of V0] (A_0) {};
    \node [below=0.5 of A_0, xshift=-5pt] (A) {$\mathcal{I}^x_{c_A}$};

    \node [right=4.5 of V0] (B_0) {};
    \node [below=0.5 of B_0, xshift=5pt] (B) {$\mathcal J^y_{c_B}$};

    \begin{scope}
    \fill[gray!30, draw=black, opacity=0.25, rounded corners=2pt]
        ($(A)+(-2,-0.05)$)
        rectangle
        ($(B)+(0.8,-0.45)$);
\end{scope}
    
    \node [below= of A] (A1) {$\{P^{xyc_A}_{a_C}\}_{a_C}$};
    \node [below= of B] (B1) {$\{Q^{xyc_B}_{b_C}\}_{b_C}$};

    \begin{scope}[]
    \fill[gray!30, draw=black, opacity=0.25, rounded corners=2pt]
        ($(A1)+(-2,-0.25)$)
        rectangle
        ($(B1)+(0.8,0.2)$);
\end{scope}

    \node [below= 5cm of V0, yshift=-8pt, xshift=12pt] (V0final) {};
    \node [below= 5cm of V1, yshift=-10pt, xshift=-10pt] (V1final) {};

    \node [right=1.2cm of V0, xshift=-10pt] (Alice) {Alice};
    \node[right= 1cm of Alice, xshift=0pt, yshift=-1pt](pos){$pos$};
    \node [right= 4.3cm  of V0] (Bob) {Bob};

    \node [above=0.3cm of V0] (posV00) {};
    \node [left=1cm of posV00] (posV0) {};
    \node [above=0.3cm of C] (posV11) {};
    \node [right=0.5cm of posV11] (posV1) {};
    \draw [->] (posV0) -- (posV1) node[midway,yshift=5pt, xshift=5pt] {position};

    \node [left=1cm of V0](t0){};
    \node [below=5.5cm of t0](t1){};
    \draw [->] (t0) -- (t1) node[midway] {time};

  \begin{scope}    \fill[gray!30, draw=black, opacity=0.25, rounded corners=2pt]        ($(A)+(-2,0.2)$)        rectangle        ($(B)+(0.8,0.5)$);\end{scope}

 \draw[->,quantum, transform canvas={xshift=4pt, yshift=0pt}] (A)--(A1){};

 \node[below=2.4cm of A](Acommit){};

 \draw[-,transform canvas={xshift=-3pt, yshift=0pt}] (A)--(Acommit){};
  \draw[-,transform canvas={xshift=-5pt, yshift=0pt}] (A)--(Acommit){};
   \draw[-,transform canvas={xshift=-7pt, yshift=0pt}] (A)--(Acommit){};
    \draw[-,transform canvas={xshift=-9pt, yshift=0pt}] (A)--(Acommit){};
     \draw[-,transform canvas={xshift=-11pt, yshift=0pt}] (A)--(Acommit){};
      \draw[-,transform canvas={xshift=-1pt, yshift=0pt}] (A)--(Acommit){};

\node [below= 4.4cm of V0, yshift=12pt, xshift=30pt] (V0com) {};
 \node[above=0.1cm of Acommit, yshift=-3pt, xshift=3.5pt](Acommit1){};
 
      \draw[->,transform canvas={xshift=-3pt, yshift=0pt}] (Acommit1)--(V0com){};
      \draw[->,transform canvas={xshift=-5pt, yshift=0pt}] (Acommit1)--(V0com){};
      \draw[->,transform canvas={xshift=-7pt, yshift=0pt}] (Acommit1)--(V0com){};
      \draw[->,transform canvas={xshift=-9pt, yshift=0pt}] (Acommit1)--(V0com){};
      \draw[->,transform canvas={xshift=-11pt, yshift=0pt}] (Acommit1)--(V0com){};
      \draw[->,transform canvas={xshift=-1pt, yshift=0pt}] (Acommit1)--(V0com){};

\draw[->,quantum, transform canvas={xshift=4pt, yshift=0pt}] (A)--(A1){};

\draw[->,quantum, transform canvas={xshift=-4pt, yshift=0pt}] (B)--(B1){};

 \node[below=2.4cm of B](Bcommit){};

 \draw[-,transform canvas={xshift=3pt, yshift=0pt}] (B)--(Bcommit){};
  \draw[-,transform canvas={xshift=5pt, yshift=0pt}] (B)--(Bcommit){};
   \draw[-,transform canvas={xshift=7pt, yshift=0pt}] (B)--(Bcommit){};
    \draw[-,transform canvas={xshift=9pt, yshift=0pt}] (B)--(Bcommit){};
     \draw[-,transform canvas={xshift=11pt, yshift=0pt}] (B)--(Bcommit){};
      \draw[-,transform canvas={xshift=1pt, yshift=0pt}] (B)--(Bcommit){};

\node [below= 4.4cm of V1, yshift=7pt, xshift=-25pt] (V1com) {};
 \node[above=0.1cm of Bcommit, yshift=-3pt, xshift=-3.5pt](Bcommit1){};
 
      \draw[->,transform canvas={xshift=3pt, yshift=0pt}] (Bcommit1)--(V1com){};
      \draw[->,transform canvas={xshift=5pt, yshift=0pt}] (Bcommit1)--(V1com){};
      \draw[->,transform canvas={xshift=7pt, yshift=0pt}] (Bcommit1)--(V1com){};
      \draw[->,transform canvas={xshift=9pt, yshift=0pt}] (Bcommit1)--(V1com){};
      \draw[->,transform canvas={xshift=11pt, yshift=0pt}] (Bcommit1)--(V1com){};
      \draw[->,transform canvas={xshift=1pt, yshift=0pt}] (Bcommit1)--(V1com){};

    \draw[->,quantum] (A)--(B1){};
    \draw[->,quantum] (B)--(A1){};

    \draw [->,transform canvas={xshift=0pt, yshift=0pt}] (A1) -- (V0final) ;
    %\draw [->,transform canvas={xshift=2pt, yshift=0pt}] (A1) -- (V0final) ;
    \draw [->,transform canvas={xshift=4pt, yshift=0pt}] (A1) -- (V0final) ;
    %\draw [->,transform canvas={xshift=-2pt, yshift=0pt}] (A1) -- (V0final) ;

    \draw [->,transform canvas={xshift=0pt, yshift=0pt}] (B1) -- (V1final) ;
    %\draw [->,transform canvas={xshift=2pt, yshift=0pt}] (B1) -- (V1final) ;
    \draw [->,transform canvas={xshift=4pt, yshift=0pt}] (B1) -- (V1final) ;
    %\draw [->,transform canvas={xshift=-2pt, yshift=0pt}] (B1) -- (V1final) ;

    \node[right = 0.1cm of V0final, fill=white, fill opacity=0.7, text opacity=1, xshift=2pt,yshift=12pt] (V0ans) {$a_C$};
    \node[left= 0.1cm of V1final, fill=white, fill opacity=0.7, text opacity=1, xshift=3pt,yshift=12pt] (V0ans) {$b_C$};

    \node[above = 1.5cm of V0final, xshift=15pt,yshift=12pt] () {$c_A$};

    \node[above = 1.5cm of V1final, xshift=-12pt,yshift=12pt] () {$c_B$};

    \end{tikzpicture}
\caption{Schematic representation of a general attack on \cQPVBBfparallel. }
\label{fig:attack-parallel_repBB84}
\end{figure}

\subsection{Normal form for attacks}

We next introduce a normal form for attacks that will be used
throughout the security proof. Put \(d=2^q\), the maximum
dimension of each attacker's initial quantum register, and
consider attacks specified by
\begin{equation}\label{eq:strategy}
\mathcal A=
\bigl(
\ket{\varphi},K,L,\mathbf{P},\mathbf{Q}
\bigr).
\end{equation}
Here, \(\ket{\varphi}_{VAB}\) is a unit vector, with
\(\dim A,\dim B\leq d\), and we write
\(\rho:=\ketbra{\varphi}{\varphi}_{VAB}\) when using
density-operator notation. The initial state is independent of the classical
inputs \(x,y\). No restriction is imposed on the marginal
state of the verifier's register~\(V\).

The family
\[
K=
\bigl\{
\bigl(K_j^x,c_A(x,j),s_A(x,j)\bigr)
\bigr\}_{x,j}
\]
specifies Alice's local operations, where, for each input \(x\),
\(j\) ranges over at most \(d^2\) values. The operator \(K_j^x\) is the single Kraus operator associated
with outcome \(j\), \(c_A(x,j)\in\{0,1\}^m\) is her announced
commitment, and \(s_A(x,j)\in\{0,\ldots,q\}\) specifies her split
of the output registers: she keeps \(s_A(x,j)\) qubits and sends
the others. These operators satisfy
\[
\sum_j(K_j^x)^\dagger K_j^x=\mathbb I_d.
\]

The complete outcome \(j\) includes all local classical records
kept by Alice. In particular, distinct outcomes may produce the
same commitment. Recording only her announced commitment gives
the instrument \(\{\mathcal I_c^x\}_c\), where
\[
\mathcal I_c^x(\tau)
=
\sum_{j:\,c_A(x,j)=c}
K_j^x\tau(K_j^x)^\dagger.
\]
For the analysis, we retain the complete local outcome in a classical
register \(J_A\):
\[
\widehat{\mathcal I}_c^x(\tau)
=
\sum_{j:\,c_A(x,j)=c}
\ketbra{j}{j}_{J_A}\otimes K_j^x\tau(K_j^x)^\dagger,
\qquad
\ptr{J_A}{\widehat{\mathcal I}_c^x(\tau)}
=
\mathcal I_c^x(\tau).
\]

Likewise,
\[
L=
\bigl\{
\bigl(L_\ell^y,c_B(y,\ell),s_B(y,\ell)\bigr)
\bigr\}_{y,\ell}
\]
specifies Bob's local operations, commitments, and splits.
For every input \(y\), the index \(\ell\) ranges over at most
\(d^2\) values, and
\[
\sum_\ell(L_\ell^y)^\dagger L_\ell^y=\mathbb I_d.
\]
His commitment instrument \(\mathcal J_c^y\) and its recorded
version \(\widehat{\mathcal J}_c^y\), retaining the complete outcome
in a classical register \(J_B\), are defined analogously.

During the communication phase, the attackers exchange their
outgoing quantum registers together with the classical pairs
\((x,j)\) and \((y,\ell)\), respectively. Classical
messages are unrestricted, and this uses only the existing simultaneous
exchange. Thus, afterwards, both attackers know
\(x,y,j,\ell\). The commitment and split associated with each outcome
can be reconstructed from the known strategy. Allowing this classical
information to be exchanged can only increase the optimal acceptance
probability, since an attack that does not use it may ignore it.

Finally,
\[
\mathbf{P}=\bigl\{P_w^{xyj\ell}\bigr\}_{x,y,j,\ell,w},
\qquad
\mathbf{Q}=\bigl\{Q_w^{xyj\ell}\bigr\}_{x,y,j,\ell,w}
\]
denote the families of Alice's and Bob's final POVM elements, respectively.
For fixed \(x,y,j,\ell\), the corresponding POVMs act on the
registers held after the exchange and may depend on \(x,y,j,\ell\) and
on the public function \(f\). Only pairs of outcomes with matching
commitments can be accepted; in that case, \(w\) is a bit string of
length \(\lvert c_A(x,j)\rvert\).

We call an attack satisfying the conditions above a \emph{normal-form
attack}. For fixed \(x,y,j,\ell\), its subnormalized state immediately
after the local instruments and before the exchange is
\begin{equation}\label{eq:joint-outcome-state}
\widetilde\rho_{j\ell}^{\,\,xy}
=
(I_V\otimes K_j^x\otimes L_\ell^y)\rho
(I_V\otimes K_j^x\otimes L_\ell^y)^\dagger.
\end{equation}
Its trace is the probability of obtaining the pair of local outcomes
\((j,\ell)\), conditional on \(x,y\). In particular,
\[
\sum_{j,\ell}
\tr{\widetilde{\rho}_{j\ell}^{\,\,xy}}
=
1.
\]

After the exchange, let
\(\widetilde\sigma_{j\ell}^{xy}\) denote the same physical state as
\(\widetilde\rho_{j\ell}^{\,\,xy}\), written with the output registers
assigned to their final recipients according to
\(s_A(x,j)\) and \(s_B(y,\ell)\). The tilde indicates that these are subnormalized states.

 For fixed \(x,y\), define
\[
\mathcal D_{xy}
=
\left\{
(j,\ell):
c_A(x,j)=c_B(y,\ell),\
\lvert c_A(x,j)\rvert\geq k
\right\}.
\]
The set \(\mathcal D_{xy}\) consists of the pairs of local
outcomes for which the attackers announce matching commitments with at
least \(k\) selected positions. All pairs outside \(\mathcal D_{xy}\)
are rejected before the final measurements are performed.

For \((j,\ell)\in\mathcal D_{xy}\), let
\[
C_{xyj\ell}
=
\left\{
i:c_A(x,j)_i=1
\right\}.
\]
Recall that for a subset \(C\subseteq\{1,\ldots,m\}\), a basis \(z\in\{0,1\}^m\), and an outcome
\(v\in\{0,1\}^{|C|}\), 
\[
M_v^{z_C}
=
\bigotimes_{i\in C}
H^{z_i}\ketbra{v_i}{v_i}H^{z_i}.
\]
For final POVMs \(P=\{P_w\}_w\) and \(Q=\{Q_w\}_w\), define the
acceptance operator
\begin{equation}\label{eq:branch-effect}
A_{C,z}(P,Q)
=
\sum_{v\in\{0,1\}^{|C|}}
M_v^{z_C}\otimes I_{V_{\overline C}}
\otimes
\sum_{\substack{w\in\{0,1\}^{|C|}\\
d_H(v,w)\leq\lfloor\gamma|C|\rfloor}}
P_w\otimes Q_w.
\end{equation}

The operator \(A_{C,z}(P,Q)\) corresponds to the accepting event on the
tested positions \(C\), when the verifier uses basis string \(z\) and
the attackers use the POVMs \(P,Q\). It aggregates all outcomes in
which the attackers give the same answer and this answer differs from
the verifier's outcome in at most \(\lfloor\gamma|C|\rfloor\) positions.

For fixed inputs \(x,y\), acceptance requires
\((j,\ell)\in\mathcal D_{xy}\). The probability of this outcome pair is
\(\Tr\widetilde\sigma_{j\ell}^{xy}\). Whenever this probability is
positive, the corresponding normalized state is
\(\widetilde\sigma_{j\ell}^{xy}/
  \Tr\widetilde\sigma_{j\ell}^{xy}\), and hence
\[
\Pr_{\mathcal A,f}[\acc\mid x,y,j,\ell]
=
\tr{
A_{C_{xyj\ell},f(x,y)}
\bigl(P^{xyj\ell},Q^{xyj\ell}\bigr)
\frac{\widetilde\sigma_{j\ell}^{xy}}
     {\Tr\widetilde\sigma_{j\ell}^{xy}}
}.
\]
Multiplying by the probability of the outcome pair cancels the
normalization factor. Outcome pairs of probability zero contribute
nothing.

Let \(N=2^n\). Summing over admissible outcome pairs and averaging
over the uniformly chosen inputs therefore gives
\begin{equation}\label{eq:acceptance-normal-form}
\begin{aligned}
\Pr_{\mathcal A,f}[\acc]
&=
\frac{1}{N^2}\sum_{x,y}
\sum_{\substack{(j,\ell)\in\mathcal D_{xy}}}
\Pr_{\mathcal A,f}[j,\ell\mid x,y]\,
\Pr_{\mathcal A,f}[\acc\mid x,y,j,\ell]\\
&=
\frac{1}{N^2}\sum_{x,y}
\sum_{(j,\ell)\in\mathcal D_{xy}}
\tr{
A_{C_{xyj\ell},f(x,y)}
\bigl(P^{xyj\ell},Q^{xyj\ell}\bigr)
\widetilde\sigma_{j\ell}^{xy}
}.
\end{aligned}
\end{equation}

The following lemma justifies restricting the security analysis
to normal-form attacks: any upper bound on their acceptance
probability also applies to general attacks in the model.
Its proof is given in the appendix.

\begin{lemma}\label{lem:strategy-form}
For every fixed \(f\), any attack in the stated model whose
local instruments have finitely many outcomes \(j\) and \(\ell\)
for each input, including all retained classical records,
is matched or outperformed by a normal-form attack of the
form \eqref{eq:strategy}, with at most \(d^2\) outcomes per
local input. For countably many such outcomes, the same
conclusion holds up to an arbitrarily small additive loss
in acceptance probability.
\end{lemma}

\section{Security proof}

The proof uses the BB84 monogamy-of-entanglement
game introduced in~\cite{TomamichelMonogamyGame2013} to bound the attackers' average success for a fixed state.
We first fix the initial state and all local instruments and
bound their average performance over random \(f\). We then
cover these choices by a finite set and apply a union bound,
so that the final guarantee also holds against attacks chosen
after \(f\) is known.

We prove the nontrivial case \(\lambda_\gamma<1\). If
\(\lambda_\gamma\geq1\), the bound in Theorem~\ref{thm:main}
follows immediately from the fact that an acceptance probability
is at most one.

For any positive semidefinite operator \(\widetilde\sigma\)
on the verifier and the two final quantum registers, define
\begin{equation}\label{eq:G}
G_C(z,\widetilde\sigma)
=
\sup_{P,Q}
\Tr\left[A_{C,z}(P,Q)\widetilde\sigma\right].
\end{equation}
The supremum ranges over local POVMs on the two final registers.
Thus \(G_C(z,\widetilde\sigma)\) is the largest probability weight
of equal, sufficiently accurate answers for the state
\(\widetilde\sigma\) and basis string \(z\).

For every choice of POVMs,
\[
0\leq
\sum_{\substack{w\in\{0,1\}^{|C|}\\
d_H(v,w)\leq\lfloor\gamma|C|\rfloor}}
P_w\otimes Q_w
\leq
\sum_{w,w'}P_w\otimes Q_{w'}
=
I\otimes I.
\]
Since \(\sum_v M_v^{z_C}=I_{V_C}\), it follows that
\begin{equation}\label{eq:acceptance-operator-bound}
0\leq A_{C,z}(P,Q)\leq I.
\end{equation}
Consequently,
\[
0\leq G_C(z,\widetilde\sigma)\leq\Tr\widetilde\sigma.
\]

Averaging this quantity over a uniformly chosen basis string
connects it to a BB84 monogamy-of-entanglement game. In this
game, the verifier measures \(V_C\) in uniformly chosen BB84
bases and reveals the bases to two separated players. The
players then perform local measurements, without further
communication. Here the winning condition requires them to
return the same string, differing from the verifier's outcome
in at most \(\lfloor\gamma|C|\rfloor\) positions.
The corresponding error-tolerant bound
from~\cite{TomamichelMonogamyGame2013} gives the following lemma.

\begin{lemma}\label{lem:monogamy}
For every positive semidefinite \(\widetilde\sigma\) and every
nonempty \(C\),
\begin{equation}\label{eq:monogamy}
2^{-m}\sum_{z\in\{0,1\}^m}G_C(z,\widetilde\sigma)
\leq\lambda_\gamma^{|C|}\Tr\widetilde\sigma.
\end{equation}
\end{lemma}

\begin{proof}
If \(\Tr\widetilde\sigma=0\), positivity gives
\(\widetilde\sigma=0\), so both sides vanish. Otherwise, put
\[
p=\Tr\widetilde\sigma>0,
\qquad
\sigma=\frac{1}{p}\Tr_{V_{\overline C}}\widetilde\sigma.
\]
Then \(\sigma\) is a normalized state on \(V_C\) and the
attackers' final registers. For every choice of POVMs,
\[
\Tr\left[A_{C,z}(P,Q)\widetilde\sigma\right]
=
p\sum_{\substack{v,w\in\{0,1\}^{|C|}\\
d_H(v,w)\leq\lfloor\gamma|C|\rfloor}}
\Tr\left[
(M_v^{z_C}\otimes P_w\otimes Q_w)\sigma
\right].
\]
In particular, the optimized expression depends on \(z\)
only through \(z_C\).

The error-tolerant monogamy bound in
\cite{TomamichelMonogamyGame2013} states that
\begin{equation}\label{eq:external-monogamy}
\begin{aligned}
&\sup_{\{P^b,Q^b\}_b}
2^{-|C|}\sum_{b\in\{0,1\}^{|C|}}
\sum_{\substack{v,w\in\{0,1\}^{|C|}\\
d_H(v,w)\leq\lfloor\gamma|C|\rfloor}}
\Tr\left[
(M_v^b\otimes P_w^b\otimes Q_w^b)\sigma
\right]\leq
\left[
2^{h(\gamma)}
\left(\frac12+\frac{1}{2\sqrt2}\right)
\right]^{|C|}
=
\lambda_\gamma^{|C|}.
\end{aligned}
\end{equation}
This bound allows arbitrary initial states and local measurements
chosen separately for every revealed basis string.

Each \(b=z_C\) has \(2^{m-|C|}\) extensions to a full string \(z\).
Moreover, the POVMs for different \(b\) can be optimized
independently. Therefore,
\[
\begin{aligned}
2^{-m}\sum_z G_C(z,\widetilde\sigma)
&=
p\,2^{-|C|}\sum_b
\sup_{P,Q}
\sum_{\substack{v,w\in\{0,1\}^{|C|}\\
d_H(v,w)\leq\lfloor\gamma|C|\rfloor}}
\Tr\left[
(M_v^b\otimes P_w\otimes Q_w)\sigma
\right]\\
&=
p\sup_{\{P^b,Q^b\}_b}
2^{-|C|}\sum_b
\sum_{\substack{v,w\in\{0,1\}^{|C|}\\
d_H(v,w)\leq\lfloor\gamma|C|\rfloor}}
\Tr\left[
(M_v^b\otimes P_w^b\otimes Q_w^b)\sigma
\right]\leq p\,\lambda_\gamma^{|C|}.
\end{aligned}
\]
Substituting \(p=\Tr\widetilde\sigma\) proves the claim.
\end{proof}

Throughout this section, let \(\epsilon=\lambda_\gamma^k\).
Write \(S=(\ket{\varphi},K,L)\) for the preparation part of
the attack \(\mathcal A\). Thus \(S\) specifies the initial state,
the local operators, and their commitment and split labels,
but does not include the final measurement families
\(\mathbf{P},\mathbf{Q}\).

For fixed inputs \(x,y\) and a basis string \(z\), define
\begin{equation}\label{eq:H}
H^S_{xy}(z)
=
\sum_{(j,\ell)\in\mathcal D_{xy}}
G_{C_{xyj\ell}}(z,\widetilde\sigma_{j\ell}^{xy}).
\end{equation}
This is the optimal acceptance probability for the preparation
strategy \(S\) at inputs \(x,y\), when the verifier uses basis
string \(z\). For each admissible outcome pair \((j,\ell)\),
the term
\(G_{C_{xyj\ell}}(z,\widetilde\sigma_{j\ell}^{xy})\)
is its largest contribution to the acceptance probability,
optimized over the final local POVMs.

The state \(\widetilde\sigma_{j\ell}^{xy}\) is subnormalized:
its trace is the probability of obtaining \((j,\ell)\), conditional
on \(x,y\). Each term in \eqref{eq:H} therefore already includes
this outcome probability. Moreover, the final POVMs may be chosen
separately for each outcome pair, since both attackers know
\(x,y,j,\ell\) after the exchange.

Using \(G_C(z,\widetilde\sigma)\leq\Tr\widetilde\sigma\)
and Lemma~\ref{lem:monogamy}, we obtain
\begin{equation}\label{eq:mean}
\begin{aligned}
0\leq H^S_{xy}(z)
&\leq
\sum_{(j,\ell)\in\mathcal D_{xy}}
\Tr\widetilde\sigma_{j\ell}^{xy}
\leq1,\\
2^{-m}\sum_z H^S_{xy}(z)
&\leq
\sum_{(j,\ell)\in\mathcal D_{xy}}
\lambda_\gamma^{|C_{xyj\ell}|}
\Tr\widetilde\sigma_{j\ell}^{xy}
\leq
\epsilon
\sum_{(j,\ell)\in\mathcal D_{xy}}
\Tr\widetilde\sigma_{j\ell}^{xy}
\leq\epsilon.
\end{aligned}
\end{equation}
The second line uses \(|C_{xyj\ell}|\geq k\) and
\(\lambda_\gamma<1\). Since the outcome probabilities sum
to at most one, summing over admissible commitments does not
introduce their number as a multiplicative factor.

For a fixed function \(f\), the basis string used at inputs
\(x,y\) is \(f(x,y)\). Averaging the optimized acceptance
probability over the uniformly chosen inputs gives
\begin{equation}\label{eq:F}
F_f(S)
=
\frac{1}{N^2}\sum_{x,y}H^S_{xy}(f(x,y)).
\end{equation}
Thus \(F_f(S)\) is the optimal overall acceptance probability
when the preparation strategy \(S\) is fixed and the final
measurement families are unrestricted. In particular, it
upper bounds the acceptance probability of any complete
attack with preparation part \(S\), as stated next.

\begin{lemma}\label{lem:continuation}
For every normal-form attack \(\mathcal A=(S,\mathbf{P},\mathbf{Q})\),
\begin{equation}\label{eq:continuation}
 \Pr_{\mathcal A,f}[\acc]\leq F_f(S).
\end{equation}
\end{lemma}
The lemma isolates the key optimization step in the proof: once the
preparation part \(S\) is fixed, no choice of final measurement families
\(\mathbf{P},\mathbf{Q}\) can
raise the acceptance probability above \(F_f(S)\). Thus, after allowing for the best possible final measurements, the relevant choices are those made before communication.

\begin{proof}
Fix \(x,y\in\{0,1\}^n\). By \eqref{eq:acceptance-normal-form},
the acceptance probability conditional on these inputs is
\[
\Pr_{\mathcal A,f}[\acc\mid x,y]
=
\sum_{(j,\ell)\in\mathcal D_{xy}}
\Tr\left[
A_{C_{xyj\ell},f(x,y)}
\bigl(P^{xyj\ell},Q^{xyj\ell}\bigr)
\widetilde\sigma_{j\ell}^{xy}
\right].
\]
For every \((j,\ell)\in\mathcal D_{xy}\), the definition of \(G_C\)
in \eqref{eq:G} gives
\[
\Tr\left[
A_{C_{xyj\ell},f(x,y)}
\bigl(P^{xyj\ell},Q^{xyj\ell}\bigr)
\widetilde\sigma_{j\ell}^{xy}
\right]
\leq
G_{C_{xyj\ell}}
\bigl(f(x,y),\widetilde\sigma_{j\ell}^{xy}\bigr).
\]
Summing over the admissible outcome pairs gives
\[
\begin{aligned}
\Pr_{\mathcal A,f}[\acc\mid x,y]
&\leq
\sum_{(j,\ell)\in\mathcal D_{xy}}
G_{C_{xyj\ell}}
\bigl(f(x,y),\widetilde\sigma_{j\ell}^{xy}\bigr)=H^S_{xy}(f(x,y)),
\end{aligned}
\]
where the equality follows from \eqref{eq:H}.
Finally, averaging over the uniformly chosen inputs and using
\eqref{eq:F}, we obtain
\[
\begin{aligned}
\Pr_{\mathcal A,f}[\acc]
&=
\frac{1}{N^2}\sum_{x,y}
\Pr_{\mathcal A,f}[\acc\mid x,y]\leq
\frac{1}{N^2}\sum_{x,y}H^S_{xy}(f(x,y))
=
F_f(S).
\end{aligned}
\]
\end{proof}

Once \(S\) is fixed independently of \(f\), the functions
\(H^S_{xy}\) are fixed, and each summand in \eqref{eq:F}
depends on \(f\) only through \(f(x,y)\). We use the following multiplicative Chernoff bound.

\begin{lemma}[Multiplicative Chernoff bound
{\cite{dubhashi2009concentration}}]
\label{lem:chernoff}
Let \(X_1,\ldots,X_T\) be independent random variables taking values
in \([0,1]\). If \(\mu>0\) satisfies
\(
\sum_{i=1}^T\mathbb E[X_i]\leq\mu,
\)
then,
\[
\Pr\left[\sum_{i=1}^T X_i>2\mu\right]
\leq
\exp\left(-\frac{\mu}{3}\right).
\]
\end{lemma}

The Chernoff bound gives the following estimate for any preparation
strategy fixed independently of \(f\).

\begin{lemma}
\label{lem:fixed-strategy}
Let \(S\) be a preparation strategy fixed independently of a
uniformly random function~\(f\). Then,
\begin{equation}\label{eq:tail}
\Pr_f[F_f(S)>2\epsilon]
\leq
\exp\left(-\frac{\epsilon N^2}{3}\right).
\end{equation}
\end{lemma}

\begin{proof}
Fix a preparation strategy \(S\) independently of the random choice
of \(f\). This fixes the initial state, the local operators, the
commitment and split labels, and hence every \(H^S_{xy}\)
in \eqref{eq:H}. For each input
pair \((x,y)\), define
\[
X_{xy}=H^S_{xy}(f(x,y)).
\]
A uniformly random function \(f\) has mutually independent truth-table
entries \(f(x,y)\), each uniformly distributed over \(\{0,1\}^m\).
Since each \(X_{xy}\) is a deterministic function of its corresponding
entry, the variables \(X_{xy}\) are mutually independent. This remains
true for pairs sharing one input, such as \((x,y)\) and \((x,y')\)
with \(y\neq y'\). The state and local operators shared by their
definitions are fixed by \(S\).
No independence of the attackers' quantum outcomes is assumed. By \eqref{eq:mean},
\[
0\leq X_{xy}\leq1,
\qquad
\mathbb E_f[X_{xy}]\leq\epsilon.
\]
There are \(N^2\) input pairs, thus
\[
\sum_{x,y}\mathbb E_f[X_{xy}]
\leq N^2\epsilon.
\]
Lemma~\ref{lem:chernoff} therefore applies with \(T=N^2\) and
\(\mu=N^2\epsilon\). Using \eqref{eq:F}, we obtain
\[
\Pr_f[F_f(S)>2\epsilon]
=
\Pr_f\left[
\sum_{x,y}X_{xy}>2N^2\epsilon
\right]
\leq
\exp\left(-\frac{\epsilon N^2}{3}\right).
\]
\end{proof}

For each preparation strategy \(S\) fixed independently of \(f\),
Lemma~\ref{lem:fixed-strategy} bounds the probability that
\(F_f(S)>2\epsilon\). However, an adversary may choose \(S\)
after learning \(f\), so we need a guarantee that holds
simultaneously for all preparation strategies.
We therefore construct a finite set of strategies independently
of \(f\), and show that every preparation strategy has an
approximation in this set whose value of \(F_f\) differs by
at most \(\epsilon/2\).

For the construction, take the registers \(A,B\) of the initial
state \(\ket{\varphi}_{VAB}\), as well as each attacker's quantum
output space before communication, to have dimension \(d\),
embedding smaller spaces when necessary. For every input \(x,y\),
index the Kraus operators \(K_j^x,L_\ell^y\) by
\(j,\ell\in\{1,\ldots,d^2\}\), adding zero operators if there
are fewer outcomes. Let \(J_A,J_B\)
be classical registers recording Alice's and Bob's respective outcomes.
Collect the local Kraus operators into the isometries
\begin{equation}\label{eq:isometry}
\begin{aligned}
W_A^x&=\sum_{j=1}^{d^2}\ket{j}_{J_A}\otimes K_j^x,
& (W_A^x)^\dagger W_A^x&=I_d,\\
W_B^y&=\sum_{\ell=1}^{d^2}\ket{\ell}_{J_B}\otimes L_\ell^y,
& (W_B^y)^\dagger W_B^y&=I_d.
\end{aligned}
\end{equation}
Each isometry maps a \(d\)-dimensional input register to the tensor
product of a \(d^2\)-dimensional outcome register and a
\(d\)-dimensional quantum output register. It is therefore represented
by a \(d^3\)-by-\(d\) matrix, with \(d^4\) complex entries, and has
operator norm one.

For any such isometry \(W=\sum_j\ket{j}\otimes K_j\), define
\[
\Phi_W(\tau)
=
\sum_j\ketbra{j}{j}\otimes K_j\tau K_j^\dagger.
\]
This channel applies the instrument and records its outcome.
In particular, put
\[
\Phi_A^x:=\Phi_{W_A^x},
\qquad
\Phi_B^y:=\Phi_{W_B^y}.
\]

Put \(\eta=\epsilon/6\) and
\(\Lambda=\ln(1+12/\epsilon)\). An \(\eta\)-net for a set is a finite
subset such that every point of the set lies within distance \(\eta\)
of a point in the subset. We use the standard volumetric covering bound:
a subset of the unit ball in a real normed space of dimension \(s\)
has an \(\eta\)-net of size at most \((1+2/\eta)^s\), with centers in
the subset; see Corollary 4.2.13 in~\cite{Vershynin_2018}.

Let \(\mathcal V\) be an \(\eta\)-net of unit vectors on the
registers \(VAB\), whose joint Hilbert space has dimension
\(D=2^m d^2\), using Euclidean distance. Let \(\mathcal W\)
be an \(\eta\)-net of the isometries in \eqref{eq:isometry},
using operator norm. Applying the estimate above with
\(\eta=\epsilon/6\) gives
\[
|\mathcal V|\leq(1+12/\epsilon)^{2D},
\qquad
|\mathcal W|\leq(1+12/\epsilon)^{2d^4},
\]
where the factors of two arise from viewing complex coordinates as real
coordinates.

We include the commitment and split labels exactly, without
approximating them. Each outcome has \(2^m\) possible commitments and
\(q+1\) possible splits: an attacker can keep any number of qubits from
\(0\) to \(q\), inclusive. Thus the set \(\mathcal L\) of label
assignments for one local input has size
\[
|\mathcal L|=[2^m(q+1)]^{d^2}.
\]
The finite set of preparation strategies \(S\) is
\[
\mathcal N
=
\mathcal V\times
(\mathcal W\times\mathcal L)^N\times
(\mathcal W\times\mathcal L)^N.
\]
This set is chosen independently of \(f\). The above
cardinality bounds give
\begin{equation}\label{eq:netcount}
\ln|\mathcal N|
\leq
2\,2^m d^2\Lambda
+
2N\left[
d^2\bigl(m\ln2+\ln(q+1)\bigr)+2d^4\Lambda
\right].
\end{equation}

Every preparation strategy \(S\) has an approximating strategy
\(S'\in\mathcal N\) with the same commitment and split labels.
Write \(\ket{\varphi'}\), \(W_A^{\prime x}\), and \(W_B^{\prime y}\)
for the initial unit vector and local isometries of \(S'\).
By construction,
\begin{equation}\label{eq:net-distances}
\|\ket{\varphi}-\ket{\varphi'}\|_2\leq\eta,
\qquad
\max_x\|W_A^x-W_A^{\prime x}\|_\infty\leq\eta,
\qquad
\max_y\|W_B^y-W_B^{\prime y}\|_\infty\leq\eta.
\end{equation}
Here \(\|\cdot\|_2\) denotes Euclidean norm and \(\|\cdot\|_\infty\)
operator norm. The following lemma shows that replacing \(S\) with its
approximation \(S'\) changes \(F_f\) by at most \(\epsilon/2\).

\begin{lemma}\label{lem:continuity}
Let \(S,S'\) have identical commitment and split labels and satisfy
\eqref{eq:net-distances}. Then, for every \(f\),
\begin{equation}\label{eq:continuity}
|F_f(S)-F_f(S')|\leq\frac{\epsilon}{2}.
\end{equation}
\end{lemma}

\begin{proof}
Write
\(\rho=\ketbra{\varphi}{\varphi}\) and
\(\rho'=\ketbra{\varphi'}{\varphi'}\).
For pure states,
\[
\|\rho-\rho'\|_1
\leq2\|\ket{\varphi}-\ket{\varphi'}\|_2
\leq2\eta.
\]
Recall that, for \(W=\sum_j\ket{j}\otimes K_j\), the channel
\(
\Phi_W(\tau)=\sum_j\ketbra{j}{j}\otimes K_j\tau K_j^\dagger
\)
applies the local instrument and records its outcome.
It acts only on the attacker's local input register; the verifier's
register \(V\) is left unchanged. In the comparisons below, the
other attacker's registers are also left unchanged. We denote all
these unaffected registers collectively by \(R\), with identity
channel \(\operatorname{id}_R\). For any density operator \(\tau\) on \(R\) and the local input
register, we first establish
\begin{equation}\label{eq:local-continuity}
\bigl\|
(\operatorname{id}_R\otimes\Phi_W)(\tau)
-
(\operatorname{id}_R\otimes\Phi_{W'})(\tau)
\bigr\|_1
\leq2\|W-W'\|_\infty.
\end{equation}
Indeed, \(\Phi_W\) applies \(W\), measures the outcome register
in the computational basis, and retains the outcome.
Contractivity under this measurement, followed by adding and
subtracting \((I_R\otimes W')\tau(I_R\otimes W)^\dagger\), gives
\[
\begin{aligned}
\bigl\|
(\operatorname{id}_R&\otimes\Phi_W)(\tau)
-
(\operatorname{id}_R\otimes\Phi_{W'})(\tau)
\bigr\|_1\leq
\bigl\|
(I_R\otimes W)\tau(I_R\otimes W)^\dagger
-
(I_R\otimes W')\tau(I_R\otimes W')^\dagger
\bigr\|_1\\
&\quad\leq
\bigl\|
(I_R\otimes(W-W'))\tau(I_R\otimes W)^\dagger
\bigr\|_1+
\bigl\|
(I_R\otimes W')\tau(I_R\otimes(W-W'))^\dagger
\bigr\|_1\\
&\quad\leq2\|W-W'\|_\infty.
\end{aligned}
\]
The last inequality uses
\(\|AXB\|_1\leq\|A\|_\infty\|X\|_1\|B\|_\infty\),
\(\|\tau\|_1=1\), and
\(\|W\|_\infty=\|W'\|_\infty=1\), since isometries preserve
vector lengths. Tensoring with \(I_R\) leaves these operator
norms unchanged.

For fixed inputs \(x,y\), let
\[
\begin{aligned}
\rho_S^{xy}
&=(\operatorname{id}_V\otimes\Phi_A^x\otimes\Phi_B^y)(\rho),\quad \text{ and }\quad
\rho_{S'}^{xy}
=(\operatorname{id}_V\otimes\Phi_A^{\prime x}
                         \otimes\Phi_B^{\prime y})(\rho'),
\end{aligned}
\]
where
\(\Phi_A^{\prime x}:=\Phi_{W_A^{\prime x}}\) and
\(\Phi_B^{\prime y}:=\Phi_{W_B^{\prime y}}\).
These normalized states include the classical outcome registers
\(J_A,J_B\). Leaving the identity channel on \(V\) implicit, the triangle
inequality gives
\[
\|\rho_S^{xy}-\rho_{S'}^{xy}\|_1
\leq
\|(\Phi_A^x\otimes\Phi_B^y)(\rho-\rho')\|_1
+
\|((\Phi_A^x-\Phi_A^{\prime x})\otimes\Phi_B^y)(\rho')\|_1
+
\|(\Phi_A^{\prime x}\otimes(\Phi_B^y-\Phi_B^{\prime y}))(\rho')\|_1.
\]
The first term is at most \(\|\rho-\rho'\|_1\), by the data-processing inequality. For the second term, apply \eqref{eq:local-continuity} to Alice's
channels after Bob's channel has been applied, treating \(V\)
and Bob's entire output as the reference register.
For the third term, apply the same estimate to Bob's channels,
treating \(V\) and Alice's entire output as the reference register.
Consequently,
\begin{equation}\label{eq:hybrid}
\begin{aligned}
\|\rho_S^{xy}-\rho_{S'}^{xy}\|_1
&\leq
\|\rho-\rho'\|_1
+2\|W_A^x-W_A^{\prime x}\|_\infty
+2\|W_B^y-W_B^{\prime y}\|_\infty\leq6\eta.
\end{aligned}
\end{equation}

After exchange, the same state, with its quantum registers
assigned to their recipients and the classical records retained,
is written as
\begin{equation}\label{eq:full-cq-state}
\sigma_S^{xy}
=
\sum_{j,\ell}
\ketbra{j}{j}_{J_A}\otimes\ketbra{\ell}{\ell}_{J_B}
\otimes\widetilde\sigma_{j\ell}^{xy},
\end{equation}
and analogously for \(S'\).
Since the split labels agree, both states undergo the same
reassignment of registers, which preserves their trace-norm
distance. Thus
\begin{equation}\label{eq:full-state-distance}
\|\sigma_S^{xy}-\sigma_{S'}^{xy}\|_1
=
\|\rho_S^{xy}-\rho_{S'}^{xy}\|_1
\leq6\eta.
\end{equation}

Fix final local POVMs
\(\{P_w^{xyj\ell}\}_w\) and
\(\{Q_w^{xyj\ell}\}_w\) for every outcome pair \(j,\ell\),
and use them for both strategies.
Let \(T\) be the POVM element corresponding to acceptance for
these fixed measurements and inputs \(x,y\), including the
classical outcome registers. The common commitment and split
labels ensure that the same \(T\) applies to both states.
Since \(0\leq T\leq I\),
\[
\left|
\Tr[T\sigma_S^{xy}]
-
\Tr[T\sigma_{S'}^{xy}]
\right|
=
\left|\Tr[T(\sigma_S^{xy}-\sigma_{S'}^{xy})]\right|
\leq
\frac12\|\sigma_S^{xy}-\sigma_{S'}^{xy}\|_1
\leq3\eta,
\]
by Lemma~\ref{lem:measurement-stability}. For any such choice of final POVMs, the acceptance probability
for \(S'\) is at most \(H^{S'}_{xy}(f(x,y))\). Therefore,
\[
\Tr[T\sigma_S^{xy}]
\leq H^{S'}_{xy}(f(x,y))+3\eta.
\]
Taking the supremum over the final POVMs, separately for each
admissible outcome pair as in \eqref{eq:H}, gives
\[
H^S_{xy}(f(x,y))
\leq H^{S'}_{xy}(f(x,y))+3\eta.
\]
Interchanging \(S\) and \(S'\) yields
\[
|H^S_{xy}(f(x,y))-H^{S'}_{xy}(f(x,y))|
\leq3\eta.
\]
Finally, averaging over the uniformly chosen inputs gives
\[
\begin{aligned}
|F_f(S)-F_f(S')|
&\leq
\frac{1}{N^2}\sum_{x,y}
|H^S_{xy}(f(x,y))-H^{S'}_{xy}(f(x,y))|\leq3\eta
=\frac{\epsilon}{2}.
\end{aligned}
\]
\end{proof}

Combining \cref{lem:fixed-strategy} for the preparation strategies in
\(\mathcal N\) with Lemma~\ref{lem:continuity} and a union bound
gives a bound that holds simultaneously for all normal-form attacks, stated in the following lemma.

\begin{lemma}\label{lem:uniform}
With probability at least
\(1-\exp[\ln|\mathcal N|-\epsilon N^2/3]\) over the choice of \(f\),
every normal-form attack \(\mathcal A\) satisfies
\[
\Pr_{\mathcal A,f}[\acc]\leq3\epsilon.
\]
\end{lemma}

\begin{proof}
Suppose a normal-form attack \(\mathcal A\) has acceptance probability
greater than \(3\epsilon\). By Lemma~\ref{lem:continuation},
\[
3\epsilon
<
\Pr_{\mathcal A,f}[\acc]
\leq F_f(S),
\]
where \(S\) is its preparation strategy. Choose \(S'\in\mathcal N\) with the same commitment and split labels
and satisfying \eqref{eq:net-distances}. By \eqref{eq:continuity},
\[
F_f(S')
\geq F_f(S)-\frac{\epsilon}{2}
>3\epsilon-\frac{\epsilon}{2}=\frac52\epsilon
>2\epsilon.
\]
Thus, whenever an attack has acceptance probability greater than
\(3\epsilon\), its preparation strategy has an approximating net
point \(S'\) satisfying \(F_f(S')>2\epsilon\). This proves the
event inclusion
\begin{equation}\label{eq:bad-event-inclusion}
\left\{
f:\exists\mathcal A,\ \Pr_{\mathcal A,f}[\acc]>3\epsilon
\right\}
\subseteq
\bigcup_{S'\in\mathcal N}
\left\{
f:F_f(S')>2\epsilon
\right\},
\end{equation}
where \(\mathcal A\) ranges over normal-form attacks.

The net is fixed independently of \(f\), so \eqref{eq:tail}
applies to each of its elements, even though the attack witnessing
the event on the left of \eqref{eq:bad-event-inclusion} may depend
on \(f\). Taking probabilities in that inclusion and applying
the union bound gives
\begin{equation}\label{eq:uniform}
\begin{aligned}
\Pr_f\left[
\exists\mathcal A:\Pr_{\mathcal A,f}[\acc]>3\epsilon
\right]
&\leq
\Pr_f\left[
\bigcup_{S'\in\mathcal N}
\{f:F_f(S')>2\epsilon\}
\right]
\leq
\sum_{S'\in\mathcal N}
\Pr_f[F_f(S')>2\epsilon]\\
&\leq
|\mathcal N|\exp\left(-\frac{\epsilon N^2}{3}\right)=
\exp\left[
\ln|\mathcal N|-\frac{\epsilon N^2}{3}
\right].
\end{aligned}
\end{equation}
The first inequality follows from \eqref{eq:bad-event-inclusion},
the second is the union bound, and the third uses \eqref{eq:tail}
for each net point. 
\end{proof}

We are now in a position to state our main result: for
\(\lambda_\gamma<1\), the probability that an attack is accepted
decays exponentially with the minimum number of required
committed positions.

\begin{theorem}\label{thm:main}
Let \(0\leq\gamma<1/2\), \(1\leq k\leq m<n\), and suppose that
\[
4q\leq n-\frac{k}{4}-\log n-7.
\]
For a uniformly random function
\(f:\{0,1\}^n\times\{0,1\}^n\to\{0,1\}^m\), with probability
at least \(1-2^{-n}\), every \emph{BE}\((q)\) attack on
\cQPVBBfparallel{} satisfies
\[
\Pr[\acc]\leq3\lambda_\gamma^k.
\]
\end{theorem}

\begin{proof}
By Lemma~\ref{lem:strategy-form}, every attack with finitely many
classical outcomes is matched or outperformed by a normal-form attack
\(\mathcal A\). For countably many outcomes, the same comparison holds
up to an arbitrarily small additive loss. It therefore suffices to
establish the claimed bound simultaneously for all normal-form attacks.

%If \(\lambda_\gamma\geq1\), the bound is immediate since an acceptance probability is at most one. We may thus assume \(\lambda_\gamma<1\) and put \(\epsilon=\lambda_\gamma^k\).
By Lemma~\ref{lem:uniform}, it remains to show that
\[
\exp\left(\ln|\mathcal N|-\frac{\epsilon N^2}{3}\right)
\leq2^{-n}.
\]
The identity \(\lambda_0^4=(17+12\sqrt2)/64>1/2\) implies
\(\lambda_\gamma\geq\lambda_0>2^{-1/4}\). Hence
\[
\epsilon\geq2^{-k/4},
\qquad
\Lambda=\ln(1+12/\epsilon)
\leq\ln(13/\epsilon)
\leq\ln13+\frac{k}{4}\ln2.
\]
The parameter condition implies \(q\leq n/4\) and
\[
N\geq128n\,2^{4q+k/4}
\geq\frac{128n\,d^4}{\epsilon}.
\]
Using \(k,m\leq n\), \(\ln(q+1)\leq q\), \(2^m\leq N\),
and \(d^2\leq d^4\) in \eqref{eq:netcount}, we obtain
\[
\begin{aligned}
\ln|\mathcal N|
&\leq
Nd^4\left[6\Lambda+2m\ln2+2\ln(q+1)\right]\leq
Nd^4\left[6\ln13+
\left(\tfrac32k+2m\right)\ln2+2q\right]\\
&\leq
Nd^4\left[6\ln13+
\left(\tfrac72\ln2+\tfrac12\right)n\right]<Nd^4(16+3n)
\leq19nNd^4
\leq\frac{19}{128}\epsilon N^2
<\frac{\epsilon N^2}{6}.
\end{aligned}
\]
Here we used \(6\ln13<16\),
\((7/2)\ln2+1/2<3\), \(n\geq1\), and \(19/128<1/6\).
Thus the failure probability in Lemma~\ref{lem:uniform} is at most
\(\exp(-\epsilon N^2/6)\). Finally,
\(\epsilon N^2\geq128nNd^4\geq6n\ln2\), so
\[
\exp\left(-\frac{\epsilon N^2}{6}\right)
\leq2^{-n}.
\]
\end{proof}

\section{Revisited sequential repetition}\label{sec:sequential-revisited}

We now revisit sequential repetition of the commitment protocol,
refining the analysis
of~\cite{allerstorfer2023makingexistingquantumposition} in two
respects. First, we account for the bias in the classical inputs
caused by conditioning on commitment. This bias was not taken into
account in the earlier expression for conditional correctness,
which averages over uniformly distributed inputs even after
conditioning; see the~\eqref{eq:seq-claimed-conditioning}.
Second, we use the same main ingredient, the gentle measurement
lemma, but apply it to a joint state instead of per every $x$ and $y$. We obtain the following bound:
\[
\Pr[\acc]\leq
\exp\left[-\frac{2k(1-\gamma-p^*)^2}{19}\right],
\]
provided \(1-\gamma>p^*\), where \(p^*\) bounds the success probability
of an attack on the original protocol without commitment.
As in the parallel protocol, \(k\) is the minimum number of
positions required to be declared conclusive by both attackers.

This result has two advantages. It bounds the overall acceptance
probability, including rejection due to inconsistent commitments,
rather than giving an alternative between detecting inconsistent
commitments and bounding the correctness of conclusive answers.
Moreover, it removes the additive \(O(k^{-1/4})\) correction to
the gap \(1-\gamma-p^*\) appearing in the earlier bound. Thus,
there is no need to accumulate enough conclusive rounds
to overcome that correction. This reduces the number of rounds
required to achieve a prescribed security level.

\subsection{Sequential experiment and one-round security hypothesis}
\label{sec:seq-model}

Instead of transmitting \(m\) qubits in parallel, the verifiers now
send them sequentially, with independently sampled classical inputs
in each round. Thus \(m\) denotes the total number of sequential
rounds, and \(1\leq k\leq m\) is the required number of conclusive
rounds. The public function is now Boolean:
\(f:\{0,1\}^n\times\{0,1\}^n\to\{0,1\}\).
The protocol is described in
Figure~\ref{fig:description_c-seq_BB84}.

\begin{figure}[htbp]
\hrule
\vspace{0.5em}
\textbf{The sequential \cQPVBBf{} protocol.}
\vspace{0.5em}
\hrule
\vspace{0.5em}
Let \(f:\{0,1\}^n\times\{0,1\}^n\to\{0,1\}\) be a public Boolean
function, \(0\leq\gamma<1/2\), and \(1\leq k\leq m\).
The following steps are performed sequentially for \(i=1,\ldots,m\).
\begin{enumerate}
\item \textbf{Preparation.}
The verifiers \Verifzero{} and \Verifone{} secretly sample independent
uniform \(x_i,y_i\in\{0,1\}^n\) and \(v_i\in\{0,1\}\), independently
of all preceding rounds. They compute \(z_i=f(x_i,y_i)\), and
\Verifzero{} prepares \(H^{z_i}\ket{v_i}\).

\item \textbf{Distribution.}
\Verifzero{} sends the qubit to \prover{} ahead of the classical
inputs. The verifiers send \(x_i\) and \(y_i\), respectively, at the
speed of light, timed to arrive simultaneously at the claimed
position \(\mathsf{pos}\). The quantum transmission is scheduled
to arrive before these inputs.

\item \textbf{Commitment.}
Before receiving the classical inputs, \prover{} broadcasts a bit
\(c_i\), equal to one if the qubit has been received and zero
otherwise. The commitment must meet the deadline prescribed for
the claimed position. If \(c_i=0\), the round is inconclusive.

\item \textbf{Measurement.}
If \(c_i=1\), upon receiving \(x_i,y_i\), \prover{} measures the
qubit in basis \(z_i=f(x_i,y_i)\) and immediately broadcasts the
outcome \(w_i\) to both verifiers.
\end{enumerate}
\textbf{Verification.}
The verifiers reject if a required message is missing or violates
the timing constraints, if their commitment bits differ in any
round, or if their answer bits differ in a conclusive round.
Otherwise let \(C=\{i\in[m]:c_i=1\}\). They accept if
\begin{equation}\label{eq:seq-verification}
|C|\geq k
\qquad\text{and}\qquad
d_H(v_C,w_C)\leq\gamma|C|.
\end{equation}
\vspace{0.5em}
\hrule
\vspace{0.2em}
\caption{The commitment protocol with \(m\) sequential transmissions.}
\label{fig:description_c-seq_BB84}
\end{figure}

An attack in each round follows the same structure as in the
parallel case, but with a single qubit and commitment and answer
bits instead of bit strings. Alice and Bob announce their commitment
bits \(c_A^i,c_B^i\) before receiving each other's input.
A commitment mismatch causes rejection. The outcome \((0,0)\)
is inconclusive, whereas \((1,1)\) is conclusive and requires
both attackers to provide an answer.

Let \(R\) count the conclusive rounds and \(W\) count those in
which both attackers answer the verifier's bit correctly.
When the required messages arrive on time and the attackers'
commitments and answers agree, condition~\eqref{eq:seq-verification}
is equivalent to
\[
R\geq k,
\qquad
W\geq(1-\gamma)R.
\]

As in the parallel case, for the security analysis we consider
the purified protocol: the verifier retains one half of an EPR
pair and measures his register in
basis \(f(x_i,y_i)\). Let \(\mathcal F_{i-1}\) denote the classical history preceding
round \(i\). Conditional on a possible history \(h\), let
\(\rho^h_{VAB}\) denote the joint state at the beginning of the
round, before its classical inputs arrive. The current inputs are uniform and independent
of this state and history.

Fix a round \(i\in\{1,\ldots,m\}\) and such a history \(h\).
For the one-round analysis, write \(x,y,\rho\) for the current
inputs and state, leaving their dependence on \(i,h\) implicit.
On receiving \(x\), Alice applies an instrument
\(\{\mathcal I_0^x,\mathcal I_1^x\}\) whose outcome is her commitment
bit; Bob similarly applies \(\{\mathcal J_0^y,\mathcal J_1^y\}\).
The outputs include the kept and sent quantum registers and any
classical records needed for subsequent operations. The attackers
then simultaneously exchange their selected registers and classical
information. When both commitments are one, they perform local
POVMs and return their answer bits. Operations in later rounds
may depend on the information and quantum memory retained by
each attacker.

Let \(M_A^x\) and \(M_B^y\) be the associated POVM elements for
Alice's and Bob's commitment outcome one, respectively.
By Lemma~\ref{lemma: instrument and channel}, the corresponding
operations can be written as
\begin{equation}\label{eq:seq-factor}
\begin{aligned}
\mathcal I_1^x(\tau)
&=\mathcal E_A^x\bigl(\sqrt{M_A^x}\tau\sqrt{M_A^x}\bigr),\\
\mathcal J_1^y(\tau)
&=\mathcal E_B^y\bigl(\sqrt{M_B^y}\tau\sqrt{M_B^y}\bigr),
\end{aligned}
\end{equation}
where \(\mathcal E_A^x,\mathcal E_B^y\) are quantum channels
with the same output registers as the corresponding instruments.

For fixed inputs, define the correctness operator
\[
\Pi^{xy}
=
\sum_{v\in\{0,1\}}
\bigl(H^{f(x,y)}\ketbra{v}{v}H^{f(x,y)}\bigr)_V
\otimes P_v^{xy}\otimes Q_v^{xy}.
\]
Here \(P_v^{xy},Q_v^{xy}\) are the POVM elements for returning
\(v\), including dependence on the available classical records.
If missing answers are allowed, these elements satisfy
\(P_0^{xy}+P_1^{xy}\leq I\) and
\(Q_0^{xy}+Q_1^{xy}\leq I\).
The event \(\mathrm{correct}\) means that both attackers return
the verifier's target bit.

For the commitment protocol, let
\(t=\Pr[c_A=c_B=1]>0\).
The probability that both attackers answer correctly, conditional
on both committing, is
\begin{equation}\label{eq:seq-commitment-success}
\begin{aligned}
p_{\mathrm{con}}
&=\Pr[\mathrm{correct}\mid c_A=c_B=1]=\frac{1}{t}\Pr[\mathrm{correct},c_A=c_B=1]\\
&=\frac{1}{t}\sum_{x,y}\Pr[x,y]\,
  \Pr[\mathrm{correct},c_A=c_B=1\mid x,y]=\frac{2^{-2n}}{t}\sum_{x,y}
\Tr\left[
\Pi^{xy}
(\operatorname{id}_V\otimes\mathcal I_1^x
                       \otimes\mathcal J_1^y)(\rho_{VAB})
\right].
\end{aligned}
\end{equation}

Replacing the instruments by quantum channels,
followed by the same exchange and final measurements, gives
an attack on the original protocol without commitment,
provided these operations satisfy its resource constraints.
For an initial state \(\sigma_{VAB}\),
the acceptance probability is
\begin{equation}\label{eq:seq-reference-operator}
\Pr[\acc]
=
2^{-2n}\sum_{x,y}
\Tr\left[
\Pi^{xy}
(\operatorname{id}_V\otimes\mathcal E_A^x
                       \otimes\mathcal E_B^y)(\sigma_{VAB})
\right].
\end{equation}
In both expressions, \(\Pi^{xy}\) acts on the corresponding
registers held by the verifier, Alice, and Bob after communication.

As in~\cite{allerstorfer2023makingexistingquantumposition}, we use an original-protocol security bound \(p^*<1\) that
applies to every input-independent density operator
\(\sigma_{VAB}\) on the allowed registers, including states
with arbitrary verifier marginals, and to every admissible
choice of local channels and final measurements. Explicitly,
\begin{equation}\label{eq:seq-reference-security}
2^{-2n}\sum_{x,y}
\Tr\left[
\Pi^{xy}
(\operatorname{id}_V\otimes\mathcal E_A^x
                       \otimes\mathcal E_B^y)(\sigma_{VAB})
\right]
\leq p^*.
\end{equation}
The same bound must apply after every preceding history. For the original protocol~\QPVBBf, bounds of \(0.98\) and approximately
\(0.8539\) were established
in~\cite{bluhm2022single} and
\cite{escolàfarràs2025quantumpositionverificationshot},
respectively, under their stated assumptions.

The following analysis connects
\eqref{eq:seq-commitment-success} and
\eqref{eq:seq-reference-operator}, transferring any such
original-protocol bound to the conditional correctness
probability of the commitment protocol.

\subsection{Conditioning on a conclusive commitment}
\label{sec:seq-conditioning}

Fix a round \(i\in\{1,\ldots,m\}\) and a possible history \(h\).
All probabilities in this subsection are conditional on \(h\),
with the history and round index omitted from the notation.
Let
\[
a=\Pr[c_A=1],
\qquad
b=\Pr[c_B=1].
\]
Recall that \(t\) is the probability that both attackers commit.
If \(t=0\), the round cannot be conclusive. Henceforth assume
\(t>0\), so that \(a,b>0\), and define
\begin{equation}\label{eq:seq-inconsistency}
\varepsilon_A
=
\Pr[c_A=0\mid c_B=1]
=\frac{b-t}{b},
\qquad
\varepsilon_B
=
\Pr[c_B=0\mid c_A=1]
=\frac{a-t}{a},
\qquad
\varepsilon=\max\{\varepsilon_A,\varepsilon_B\}.
\end{equation}
The inconsistency parameter \(\varepsilon\) may vary with the
round and history. To see how conditioning affects the input distribution, let
\[
p_{xy}=\Pr[c_A=c_B=1\mid x,y],
\]
then \(t=2^{-2n}\sum_{x,y}p_{xy}\), and Bayes' rule gives
\begin{equation}\label{eq:seq-input-bias}
\Pr[x,y\mid c_A=c_B=1]
=
\frac{2^{-2n}p_{xy}}{t}.
\end{equation}
In the following sums, we omit input pairs with \(p_{xy}=0\),
since they have zero probability conditional on both attackers
committing. The conditional correctness probability is
\[
\begin{aligned}
p_{\mathrm{con}}
&=
\sum_{x,y}
\frac{2^{-2n}p_{xy}}{t}\,
\Pr[\mathrm{correct}\mid x,y,c_A=c_B=1]\\
&=
2^{-2n}\sum_{x,y}
\frac{p_{xy}}{t}\,
\Tr\left[
\Pi^{xy}
\frac{(\operatorname{id}_V\otimes\mathcal I_1^x
                         \otimes\mathcal J_1^y)(\rho)}{p_{xy}}
\right].
\end{aligned}
\]
The earlier analysis
of~\cite{allerstorfer2023makingexistingquantumposition}
expressed this probability, in the present notation, as
\begin{equation}\label{eq:seq-claimed-conditioning}
p_{\mathrm{con}}
\stackrel{\text{claimed}}{=}
2^{-2n}\sum_{x,y}
\Tr\left[
\Pi^{xy}
\frac{(\operatorname{id}_V\otimes\mathcal I_1^x
                         \otimes\mathcal J_1^y)(\rho)}{p_{xy}}
\right].
\end{equation}
The state inside each summand is normalized separately for its
input pair and is identical in both expressions. The difference
lies in the weights: by \eqref{eq:seq-input-bias}, conditioning
on both attackers committing gives input probabilities
\(2^{-2n}p_{xy}/t\), rather than \(2^{-2n}\).
For example, if Alice commits exactly when \(x\) belongs to a
fixed nonempty proper subset and Bob always commits, conditioning
on both committing restricts \(x\) to that subset.

To keep track of both the conditional input distribution and
the corresponding quantum states, we introduce classical
registers \(X,Y\) and define
\[
U_X=2^{-n}\sum_x\ketbra{x}{x},
\qquad
U_Y=2^{-n}\sum_y\ketbra{y}{y},
\qquad
\Omega=U_X\otimes U_Y\otimes\rho_{VAB}.
\]
The registers \(X,Y\) represent the classical inputs, not
additional quantum memory for the attackers. The state
\(\Omega\) describes the uniform inputs and the independent
quantum state before conditioning on commitment.
On these joint registers, define the POVM elements
\[
\begin{aligned}
M_A&=\sum_x\ketbra{x}{x}_X\otimes I_Y\otimes I_V
                  \otimes M_A^x\otimes I_B,\\
M_B&=\sum_y I_X\otimes\ketbra{y}{y}_Y\otimes I_V
                  \otimes I_A\otimes M_B^y.
\end{aligned}
\]
They satisfy \(0\leq M_A,M_B\leq I\) and commute, since they
are nontrivial on disjoint registers.
Because \(M_A^x\) is the POVM element associated with Alice's
commitment outcome one,
\[
\begin{aligned}
\Tr(M_A\Omega)
&=2^{-n}\sum_x
\Tr\left[
(I_V\otimes M_A^x\otimes I_B)\rho_{VAB}
\right]=2^{-n}\sum_x\Pr[c_A=1\mid x]
=a.
\end{aligned}
\]
Analogously,
\[
\Tr(M_B\Omega)=b,
\qquad
\Tr(M_AM_B\Omega)=t.
\]
The normalized states associated with these commitment outcomes,
before applying the channels in \eqref{eq:seq-factor}, are
\begin{equation}\label{eq:seq-filtered-states}
\begin{aligned}
\omega_A&=\frac{\sqrt{M_A}\Omega\sqrt{M_A}}{a},
\qquad
\omega_B=\frac{\sqrt{M_B}\Omega\sqrt{M_B}}{b},\qquad
\omega_{AB}=\frac{\sqrt{M_A}\sqrt{M_B}\Omega
             \sqrt{M_B}\sqrt{M_A}}{t}.
\end{aligned}
\end{equation}
All three are states on the registers \(XYVAB\).
The subscripts \(A,B,AB\) indicate whose commitment outcomes
have been selected; they do not denote reduced states. For given $x$ and $y$, write
\[
D_{xy}=I_V\otimes\sqrt{M_A^x}\otimes\sqrt{M_B^y}.
\]
Then \(p_{xy}=\Tr(D_{xy}\rho D_{xy})\), and, omitting input
pairs with \(p_{xy}=0\), since they contribute nothing, 
\[
\omega_{AB}
=
\sum_{x,y}
\frac{2^{-2n}p_{xy}}{t}\,
\ketbra{xy}{xy}_{XY}
\otimes\frac{D_{xy}\rho D_{xy}}{p_{xy}}.
\]
The weight \(2^{-2n}p_{xy}/t\) is the conditional probability
of the input pair, and \(D_{xy}\rho D_{xy}/p_{xy}\) is its
corresponding normalized quantum state before applying the
channels. In particular, for every input pair,
\[
\langle xy|\omega_{AB}|xy\rangle_{XY}
=
\frac{2^{-2n}}{t}D_{xy}\rho D_{xy}.
\]
Substituting \eqref{eq:seq-factor} into
\eqref{eq:seq-commitment-success} therefore gives
\begin{equation}\label{eq:seq-conditional-success}
\begin{aligned}
p_{\mathrm{con}}
&=\frac{2^{-2n}}{t}\sum_{x,y}
\Tr\left[
\Pi^{xy}
(\operatorname{id}_V\otimes\mathcal E_A^x\otimes\mathcal E_B^y)
(D_{xy}\rho D_{xy})
\right]\\
&=\sum_{x,y}
\Tr\left[
\Pi^{xy}
(\operatorname{id}_V\otimes\mathcal E_A^x\otimes\mathcal E_B^y)
\bigl(\langle xy|\omega_{AB}|xy\rangle_{XY}\bigr)
\right].
\end{aligned}
\end{equation}
Thus the same channels and correctness test used for the original
protocol now act on the classical inputs and quantum state
described by \(\omega_{AB}\).

As in~\cite{allerstorfer2023makingexistingquantumposition}, the key
tool is the gentle measurement lemma. The earlier argument compares
states conditioned on individual input pairs and uses a bipartite
graph argument to identify a large subset of pairs with nearby
conditional states. Here we instead retain the classical inputs
in the joint state and apply the gentle measurement lemma directly
to these classical--quantum states. This treats all input pairs
together, with their actual conditional probabilities.

\begin{lemma}\label{lem:seq-one-round}
Under the original-protocol bound \eqref{eq:seq-reference-security},
for every history and round with \(t>0\),
\begin{equation}\label{eq:seq-one-round}
p_{\mathrm{con}}\leq p^*+3\sqrt{\varepsilon}.
\end{equation}
\end{lemma}

This estimate gives a simpler connection between commitment
inconsistency and conditional correctness than the corresponding
analysis in~\cite{allerstorfer2023makingexistingquantumposition}.
This bound depends on a single parameter, \(\varepsilon\), rather
than two parameters controlling the inconsistency on a subset
of input pairs and the size of its complement.
By treating the classical inputs together with the quantum state,
we avoid introducing these subsets or using a bipartite graph
argument.

\begin{proof}
Consider the measurement \(\{M_B,I-M_B\}\) on \(\omega_A\).
The first outcome has probability
\[
\Tr(M_B\omega_A)=\frac{t}{a}=1-\varepsilon_B,
\]
and its normalized post-measurement state is
\[
\frac{\sqrt{M_B}\omega_A\sqrt{M_B}}
     {\Tr(M_B\omega_A)}
=\omega_{AB}.
\]
The gentle measurement lemma
(Lemma~\ref{lem:seq-gentle}) therefore gives
\begin{equation}\label{eq:seq-gentle-comparison}
\|\omega_{AB}-\omega_A\|_1
\leq2\sqrt{\varepsilon_B}.
\end{equation}

Similarly, applying the measurement \(\{M_A,I-M_A\}\) to
\(\omega_B\) gives the first outcome with probability
\[
\Tr(M_A\omega_B)=\frac{t}{b}=1-\varepsilon_A,
\]
and normalized post-measurement state
\[
\frac{\sqrt{M_A}\omega_B\sqrt{M_A}}
     {\Tr(M_A\omega_B)}
=\omega_{AB}.
\]
A second application of the gentle measurement lemma gives
\begin{equation}\label{eq:seq-gentle-comparison-B}
\|\omega_{AB}-\omega_B\|_1
\leq2\sqrt{\varepsilon_A}.
\end{equation}

Recall that \(X,Y\) are the classical registers storing the inputs
\(x,y\). Let \(\mathcal R_X\) discard the contents of \(X\) and
replace them with an independent uniform input:
\[
\mathcal R_X(\zeta)=U_X\otimes\Tr_X\zeta.
\]
This is a quantum channel. Since Bob's commitment measurement
depends only on \(y\), register \(X\) remains uniform and independent
of the other registers in \(\omega_B\). Hence
\[
\mathcal R_X(\omega_B)=\omega_B.
\]
Likewise, Alice's commitment measurement depends only on \(x\),
so \(Y\) remains uniform and independent in \(\omega_A\).
Tracing out \(X\) averages over Alice's inputs, giving
\begin{equation}\label{eq:seq-reference-state}
\begin{aligned}
\mathcal R_X(\omega_A)
&=U_X\otimes U_Y\otimes\sigma_A,\\
\sigma_A
&=
\frac{2^{-n}}{a}\sum_x
(I_V\otimes\sqrt{M_A^x}\otimes I_B)\rho
(I_V\otimes\sqrt{M_A^x}\otimes I_B).
\end{aligned}
\end{equation}
The state \(\sigma_A\) has trace one and is defined on the
original pre-instrument registers \(VAB\). In
\(\mathcal R_X(\omega_A)\), it is independent of both uniform
classical inputs.

Using \(\mathcal R_X(\omega_B)=\omega_B\), we obtain
\begin{equation}\label{eq:seq-reference-distance}
\begin{aligned}
\|\omega_{AB}-\mathcal R_X(\omega_A)\|_1
&\leq
\|\omega_{AB}-\omega_B\|_1
+\|\mathcal R_X(\omega_B)-\mathcal R_X(\omega_A)\|_1\leq
\|\omega_{AB}-\omega_B\|_1+\|\omega_B-\omega_A\|_1\\
&\leq
2\|\omega_{AB}-\omega_B\|_1
+\|\omega_{AB}-\omega_A\|_1\leq4\sqrt{\varepsilon_A}+2\sqrt{\varepsilon_B}\\&\leq 6\sqrt{\varepsilon}.
\end{aligned}
\end{equation}
The first and third inequalities are triangle inequalities.
The second is the data-processing inequality for trace distance
under the channel \(\mathcal R_X\), the fourth follows from
\eqref{eq:seq-gentle-comparison} and
\eqref{eq:seq-gentle-comparison-B},  and the last uses
\(\varepsilon=\max\{\varepsilon_A,\varepsilon_B\}\).

We apply the same channels, exchange of registers, and final
measurements to \(\omega_{AB}\) and
\(\mathcal R_X(\omega_A)\),
using the classical inputs stored in \(X,Y\). For \(\omega_{AB}\), the probability of correct answers is,
by \eqref{eq:seq-conditional-success},
\[
p_{\mathrm{con}}
=
\sum_{x,y}
\Tr\left[
\Pi^{xy}
(\operatorname{id}_V\otimes\mathcal E_A^x\otimes\mathcal E_B^y)
\bigl(\langle xy|\omega_{AB}|xy\rangle_{XY}\bigr)
\right].
\]
For \(\mathcal R_X(\omega_A)\), note that
\[
\langle xy|\mathcal R_X(\omega_A)|xy\rangle_{XY}
=2^{-2n}\sigma_A.
\]
Thus the corresponding probability is
\[
\begin{aligned}
p_{\mathrm{orig}}
&:=
\sum_{x,y}
\Tr\left[
\Pi^{xy}
(\operatorname{id}_V\otimes\mathcal E_A^x\otimes\mathcal E_B^y)
\bigl(\langle xy|\mathcal R_X(\omega_A)|xy\rangle_{XY}\bigr)
\right]=
2^{-2n}\sum_{x,y}
\Tr\left[
\Pi^{xy}
(\operatorname{id}_V\otimes\mathcal E_A^x\otimes\mathcal E_B^y)
(\sigma_A)
\right]
\leq p^*,
\end{aligned}
\]
where the last inequality is
\eqref{eq:seq-reference-security}. The channels and final correctness test are identical for the
two states. Applying
Lemma~\ref{lem:measurement-stability} therefore gives
\[
\begin{aligned}
|p_{\mathrm{con}}-p_{\mathrm{orig}}|
&=
\left|
\sum_{x,y}
\Tr\left[
\Pi^{xy}
(\operatorname{id}_V\otimes\mathcal E_A^x\otimes\mathcal E_B^y)
\bigl(
\langle xy|
(\omega_{AB}-\mathcal R_X(\omega_A))
|xy\rangle_{XY}
\bigr)
\right]
\right|\leq
\frac12\|\omega_{AB}-\mathcal R_X(\omega_A)\|_1.
\end{aligned}
\]
Then,
\[
\begin{aligned}
p_{\mathrm{con}}
&=p_{\mathrm{orig}}
  +p_{\mathrm{con}}-p_{\mathrm{orig}}
\leq p_{\mathrm{orig}}
  +|p_{\mathrm{con}}-p_{\mathrm{orig}}|\leq p^*
  +\frac12\|\omega_{AB}-\mathcal R_X(\omega_A)\|_1\leq p^*+3\sqrt{\varepsilon},
\end{aligned}
\]
here the last inequality uses
\eqref{eq:seq-reference-distance}.
\end{proof}

\subsection{Combining correctness and commitment consistency}
\label{sec:seq-proof}

We now combine the one-round estimate of
Lemma~\ref{lem:seq-one-round} with a statistical analysis of
the sequential execution. By accounting for both incorrect
answers and rejection due to inconsistent commitments, we obtain
the following unconditional bound on the acceptance probability.

\begin{theorem}\label{thm:seq-soundness}
Suppose the original-protocol bound \eqref{eq:seq-reference-security}
holds with the same \(p^*\) after every possible history.
For every fixed number \(m\geq k\geq1\) of sequential attempts and
every \(0\leq\gamma<1-p^*\), any adaptive attack has acceptance
probability at most
\begin{equation}\label{eq:seq-soundness}
\Pr[\acc]
\leq
\exp\left[-\frac{2k(1-\gamma-p^*)^2}{19}\right].
\end{equation}
\end{theorem}

This is an unconditional acceptance bound: it includes rejection due
to inconsistent commitments and is not conditioned on avoiding such
rejection.

\begin{proof}
Put \(\Delta=1-\gamma-p^*>0\).
We construct nonnegative weights whose expected values do not
increase from one round to the next. These weights are large
whenever the verifiers accept, allowing us to bound the acceptance
probability without assuming independence between rounds.

Fix a round \(i\), condition on \(\mathcal F_{i-1}\), and first
suppose that no earlier commitment mismatch has occurred.
Denote the conditional commitment probabilities by
\[
\begin{aligned}
t&=\Pr[(c_A^i,c_B^i)=(1,1)\mid\mathcal F_{i-1}],&
u&=\Pr[(c_A^i,c_B^i)=(1,0)\mid\mathcal F_{i-1}],\\
v&=\Pr[(c_A^i,c_B^i)=(0,1)\mid\mathcal F_{i-1}],&
z&=\Pr[(c_A^i,c_B^i)=(0,0)\mid\mathcal F_{i-1}].
\end{aligned}
\]
Thus \(t+u+v+z=1\), and \(u+v\) is the probability of a
commitment mismatch. If \(t>0\), the inconsistency parameter
for this round and history is
\[
\varepsilon
=
\max\left\{
\frac{v}{t+v},\frac{u}{t+u}
\right\}.
\]
Writing \(p_{\mathrm{con}}\) for the conditional probability
that both answers are correct given a conclusive commitment,
Lemma~\ref{lem:seq-one-round} gives
\[
p_{\mathrm{con}}\leq p^*+3\sqrt{\varepsilon}.
\]
We also need to relate \(\varepsilon\) to the probability of
a commitment mismatch. Its definition gives
\[
\begin{aligned}
1-\varepsilon
&=
\min\left\{\frac{t}{t+v},\frac{t}{t+u}\right\}
=\frac{t}{t+\max\{u,v\}}\geq\frac{t}{t+u+v}.
\end{aligned}
\]
Thus,
\begin{equation}\label{eq:seq-commitment-risk}
t\leq(t+u+v)(1-\varepsilon).
\end{equation}
Then the possible increase \(3\sqrt{\varepsilon}\) in conditional
correctness comes with a risk of rejection: among rounds in which
at least one attacker commits, the conditional probability of
a mismatch is at least~\(\varepsilon\).

For \(j=0,\ldots,m\), let \(R_j\) count the conclusive rounds
among the first \(j\), and let \(W_j\) count those in which both
answers are correct. In particular, \(R_m=R\) and \(W_m=W\), where \(R,W\) are
the total counts introduced earlier.
Let \(\mathsf S_j\) be the event that no commitment mismatch
has occurred in those rounds. If the protocol stops before round \(m\), we treat all remaining rounds as inconclusive for the analysis, so these quantities
remain defined through round \(m\).
Define
\[
A_j=W_j-(1-\gamma)R_j.
\]
The quantity \(A_j\) is nonnegative exactly when
\(W_j\geq(1-\gamma)R_j\), the required correctness condition for the conclusive rounds among the first \(j\) rounds. For positive parameters \(\theta,\kappa\), to be chosen below,
define
\begin{equation}\label{eq:seq-weight}
Z_j=\mathbf1_{\mathsf S_j}
\exp\left(\theta A_j+\kappa R_j\right).
\end{equation}
The indicator makes this weight zero after a commitment mismatch.
If the verifiers accept, no commitment mismatch has occurred,
\(A_m\geq0\), and \(R_m\geq k\). Since \(\theta,\kappa>0\),
\[
Z_m
=\exp(\theta A_m+\kappa R_m)
\geq\exp(\kappa R_m)
\geq e^{\kappa k}.
\]

We next examine how each commitment outcome affects the weight \(Z_j\).
An inconclusive outcome changes neither \(A_j\) nor \(R_j\),
so it leaves the weight unchanged. A mismatch makes the weight
zero. In a conclusive round, let \(Y_i=1\) if both answers are
correct and \(Y_i=0\) otherwise. Then
\[
R_i-R_{i-1}=1,
\qquad
A_i-A_{i-1}=Y_i-(1-\gamma).
\]
Consequently, on a history with no previous mismatch,
\begin{equation}\label{eq:seq-multiplier}
\frac{Z_i}{Z_{i-1}}
=
\begin{cases}
1,
& (c_A^i,c_B^i)=(0,0),\\
0,
& (c_A^i,c_B^i)=(1,0)\text{ or }(0,1),\\
\exp\left(\kappa+\theta[Y_i-(1-\gamma)]\right),
& (c_A^i,c_B^i)=(1,1).
\end{cases}
\end{equation}
In the last case, the multiplier is \(e^{\kappa+\theta\gamma}\)
when both answers are correct and
\(e^{\kappa-\theta(1-\gamma)}\) otherwise.

We next bound the conditional expectation of the ratio
\(Z_i/Z_{i-1}\) in a conclusive round. Conditional on \((c_A^i,c_B^i)=(1,1)\) and \(\mathcal F_{i-1}\),
the binary variable \(Y_i\) has mean \(p_{\mathrm{con}}\).
Hoeffding's lemma (Lemma~\ref{lem:seq-hoeffding}) gives
\[
\mathbb E\left[
e^{\theta(Y_i-p_{\mathrm{con}})}
\,\middle|\,
(c_A^i,c_B^i)=(1,1),\mathcal F_{i-1}
\right]
\leq e^{\theta^2/8}.
\]
Writing
\(Y_i-(1-\gamma)=(Y_i-p_{\mathrm{con}})
+[p_{\mathrm{con}}-(1-\gamma)]\), we obtain
\[
\begin{aligned}
\mathbb E&\left[
e^{\theta[Y_i-(1-\gamma)]}
\,\middle|\,
(c_A^i,c_B^i)=(1,1),\mathcal F_{i-1}
\right]=
e^{\theta[p_{\mathrm{con}}-(1-\gamma)]}
\mathbb E\left[
e^{\theta(Y_i-p_{\mathrm{con}})}
\,\middle|\,
(c_A^i,c_B^i)=(1,1),\mathcal F_{i-1}
\right]\\
&\quad\leq
\exp\left(
\theta[p_{\mathrm{con}}-(1-\gamma)]
+\frac{\theta^2}{8}
\right)\leq
\exp\left(
-\theta\Delta+3\theta\sqrt{\varepsilon}
+\frac{\theta^2}{8}
\right).
\end{aligned}
\]
The last inequality uses
\(p_{\mathrm{con}}\leq p^*+3\sqrt{\varepsilon}\) and
\(\Delta=1-\gamma-p^*\).

We now average over the commitment outcomes in
\eqref{eq:seq-multiplier}. In the resulting conditional expectation,
the inconclusive outcome contributes \(z\cdot1=z\), the mismatches
contribute \((u+v)\cdot0=0\), and the conditional expectation for
a conclusive outcome is weighted by its probability \(t\). Therefore,
\[
\begin{aligned}
\mathbb E\left[
\frac{Z_i}{Z_{i-1}}
\,\middle|\,\mathcal F_{i-1}
\right]
&=
z+t e^\kappa
\mathbb E \left[
e^{\theta[Y_i-(1-\gamma)]}
\,\middle|\,
(c_A^i,c_B^i)=(1,1),\mathcal F_{i-1}
\right]
\leq
z+t\exp\left(
\kappa-\theta\Delta+3\theta\sqrt{\varepsilon}
+\frac{\theta^2}{8}
\right)\\
&\leq
z+(t+u+v)\exp\left(
\kappa-\theta\Delta-\varepsilon
+3\theta\sqrt{\varepsilon}
+\frac{\theta^2}{8}
\right)
\leq
z+(t+u+v)\exp\left(
\kappa-\theta\Delta+\frac{19\theta^2}{8}
\right).
\end{aligned}
\]
In the second inequality we use~\eqref{eq:seq-commitment-risk}
and \(1-\varepsilon\leq e^{-\varepsilon}\), and for the third inequality, we bound the terms in the exponent as
\[
\begin{aligned}
-\varepsilon+3\theta\sqrt{\varepsilon}+\frac{\theta^2}{8}
&=
-\left(\sqrt{\varepsilon}-\frac{3\theta}{2}\right)^2
+\frac{9\theta^2}{4}+\frac{\theta^2}{8}\leq
\frac{9\theta^2}{4}+\frac{\theta^2}{8}
=\frac{19\theta^2}{8}.
\end{aligned}
\]
Then, the resulting estimate no longer depends on the
round-dependent inconsistency parameter~\(\varepsilon\). 
Choose
\[
\theta=\frac{4\Delta}{19},
\qquad
\kappa=\frac{2\Delta^2}{19}.
\]
These parameters are positive and satisfy
\(
\kappa-\theta\Delta+\frac{19\theta^2}{8}=0.
\)
Consequently,
\[
\mathbb E\left[
\frac{Z_i}{Z_{i-1}}
\,\middle|\,\mathcal F_{i-1}
\right]
\leq z+(t+u+v)e^0=1.
\]
If \(t=0\), a conclusive round cannot occur, and the expected
multiplier is simply \(z\leq1\); no conditional correctness
probability or inconsistency parameter is needed in that case.

On a history with no earlier mismatch, \(Z_{i-1}>0\).
Since \(Z_{i-1}\) is fixed when conditioning on
\(\mathcal F_{i-1}\),
\[
\begin{aligned}
\mathbb E[Z_i\mid\mathcal F_{i-1}]
&=
Z_{i-1}
\mathbb E\left[
\frac{Z_i}{Z_{i-1}}
\,\middle|\,\mathcal F_{i-1}
\right]\leq Z_{i-1}.
\end{aligned}
\]
If an earlier mismatch has occurred, then
\(Z_{i-1}=Z_i=0\), so the same inequality holds.
Thus \((Z_j)_{j=0}^m\) is a nonnegative supermartingale.
Taking expectations and iterating gives
\[
\mathbb E[Z_m]\leq\mathbb E[Z_{m-1}]
\leq\cdots\leq Z_0=1.
\]

Finally, every accepting execution has no commitment mismatch,
at least \(k\) conclusive rounds, and
\(W_m\geq(1-\gamma)R_m\). Therefore, on acceptance,
\[
Z_m
=\exp(\theta A_m+\kappa R_m)
\geq e^{\kappa k}.
\]
On a rejecting execution, the indicator
\(\mathbf1_{\{\acc\}}\) is zero, while \(Z_m\geq0\).
Thus, in either case,
\[
\mathbf1_{\{\acc\}}e^{\kappa k}\leq Z_m.
\]
Taking expectations yields
\[
e^{\kappa k}\Pr[\acc]
\leq\mathbb E[Z_m]\leq1.
\]
Rearranging and substituting
\(\kappa=2\Delta^2/19\), we obtain
\[
\Pr[\acc]
\leq
\exp\left[-\frac{2k(1-\gamma-p^*)^2}{19}\right].
\]
\end{proof}

\section{Discussion and acknowledgments}

We have analyzed the commitment modification of
\QPVBBfparallel{} and revisited the sequential commitment
protocol. The parallel proof treats the local quantum
instruments explicitly while retaining the reduction to a
monogamy-of-entanglement game and the counting argument.
For sequential repetition, accounting for the conditional
input distribution leads to an unconditional acceptance bound
without a finite-size correction to the correctness gap.

An important direction for future work is to establish
comparable parallel security guarantees for explicit,
efficiently computable functions \(f\). Our parallel result
holds with high probability for a uniformly random function,
whereas practical implementations call for concrete choices
with manageable computational costs.

Another direction is to prove a parallel-repetition bound
for the original protocol with basis function \(f^{\otimes m}\),
where \(f\) is Boolean and is applied to a separate input pair
in each coordinate. One could then investigate whether this
bound can be extended to the parallel commitment protocol.

\emph{Acknowledgements.}
We thank Philip Verduyn Lunel and Andreas Bluhm for their
helpful feedback on an earlier version of this manuscript.
The role of ChatGPT in developing the revised arguments
is described in the introduction. 

This work was supported by the Dutch Ministry of Economic
Affairs and Climate Policy (EZK), as part of the Quantum
Delta NL program, and by the project Divide and Quantum
`D\&Q' NWA.1389.20.241 of the program `NWA-ORC', which is
partly funded by the Dutch Research Council (NWO).

%========================================================================================
 
\bibliographystyle{alphaurl}
\bibliography{biblio.bib}

\newpage

\appendix

\section{Justification of the strategy representation}
\label{app:strategy-form}

\begin{proof}[Proof of Lemma~\ref{lem:strategy-form}]
Fix a public function \(f\) and an arbitrary attack. We first treat
finitely many classical outcomes and show that a normal-form attack
performs at least as well. We then handle countably many outcomes
by approximation.

Each attacker's quantum input register and its combined kept and sent
output registers have dimension at most \(d=2^q\). Embed them in
\(q\)-qubit registers and extend the local operations on unused
subspaces. This leaves the original measurement probabilities unchanged;
it imposes no additional bound on the registers held after the exchange. 

Any shared randomness generated before the inputs \(x,y\) are received
expresses the attack as a convex combination of attacks with fixed
classical data. The acceptance probability of the original attack is
the corresponding convex combination of their acceptance
probabilities. Hence, one of these fixed choices performs at least as
well as the original attack, and we may fix it.

Now fix all local instruments and all later operations of the attack.
For an initial state \(\rho\), let \(p(\rho)\) denote the resulting
probability that the verifiers accept. If
\[
\rho=\sum_a \lambda_a
\ketbra{\varphi_a}{\varphi_a}
\]
is a spectral decomposition of \(\rho\), where each
\(\ket{\varphi_a}\) is a unit vector, then
\[
p(\rho)
=
\sum_a \lambda_a\,
p\left(\ketbra{\varphi_a}{\varphi_a}\right)
\leq
\max_a
p\left(\ketbra{\varphi_a}{\varphi_a}\right).
\]
Consequently, there is a pure state that performs at least as well as
\(\rho\). We may therefore take
\[
\rho=\ketbra{\varphi}{\varphi}_{VAB}
\]
for a unit vector \(\ket{\varphi}_{VAB}\).

%\medskip

%\noindent\textbf{Refining the local outcomes and bounding their number.}
Consider Alice at a fixed input \(x\). Let \(u\) denote her complete
original classical outcome, including the commitment and every other
classical record produced at this stage. Her operation has a Kraus
representation
\[
\mathcal I_u^x(\tau)
=
\sum_a T_{u,a}^x\tau(T_{u,a}^x)^\dagger.
\]
Alice may additionally record the Kraus index \(a\), replacing this
operation by
\[
\tau\longmapsto
\sum_a
\ketbra{u,a}{u,a}
\otimes T_{u,a}^x\tau(T_{u,a}^x)^\dagger.
\]
This is a valid instrument because
\[
\sum_{u,a}(T_{u,a}^x)^\dagger T_{u,a}^x=\mathbb I.
\]
Moreover, erasing the additional register recovers the original
operation:
\[
\Tr_a\left[
\sum_a
\ketbra{u,a}{u,a}
\otimes T_{u,a}^x\tau(T_{u,a}^x)^\dagger
\right]
=
\ketbra{u}{u}\otimes\mathcal I_u^x(\tau).
\]
Thus, if all later operations ignore \(a\), the acceptance probability
is unchanged. Relabel the complete outcome \((u,a)\) by a single index
\(j\), and write
\[
K_j^x:=T_{u,a}^x.
\]
The commitment \(c_A(x,j)\) and split \(s_A(x,j)\) attached to \(j\)
are those originally attached to \(u\). In particular, this refinement
retains every pre-existing classical record. Apply the same construction
to Bob for every input \(y\), writing his resulting Kraus operators as
\(L_\ell^y\).

We now bound the number of outcomes. Fix \(x\), and hold the initial
state, all other local instruments, and all final measurements fixed.
If Alice has more than \(d^2\) nonzero operators \(K_j^x\), then their
associated POVM elements
\[
M_j^x:=(K_j^x)^\dagger K_j^x
\]
are linearly dependent in the real \(d^2\)-dimensional vector space of
Hermitian operators on \(\mathbb C^d\). Therefore, there are real
numbers \(\alpha_j\), not all zero, such that
\[
\sum_j\alpha_jM_j^x=0.
\]
For a real parameter \(t\), replace
\[
K_j^x
\longmapsto
\sqrt{1+t\alpha_j}\,K_j^x.
\]
Whenever \(1+t\alpha_j\geq0\) for every \(j\), these operators define a
valid instrument. Indeed,
\[
\sum_j(1+t\alpha_j)M_j^x
=
\sum_jM_j^x+t\sum_j\alpha_jM_j^x
=
\mathbb I_d.
\]

This perturbative argument is closely related to the extremality
arguments used in~\cite{dariano2011extremal}. We include the
details here because the present setting must preserve the classical
labels on which the later communication and measurements may depend.

For every \(j\), let \(b_j\) be the contribution to the acceptance
probability from Alice's outcome \(j\) at this fixed input \(x\):
\[
b_j:=
\frac{1}{N^2}\sum_y
\sum_{\ell:(j,\ell)\in\mathcal D_{xy}}
\Tr\left[
A_{C_{xyj\ell},f(x,y)}
(P^{xyj\ell},Q^{xyj\ell})
\widetilde\sigma_{j\ell}^{xy}
\right].
\]
Here \(\widetilde\sigma_{j\ell}^{xy}\) denotes the subnormalized state
after applying \(K_j^x,L_\ell^y\) and assigning their output registers
to the final recipients. Thus it acts on the same registers as the
acceptance operator, and its trace includes the probability of the
outcome pair. In particular, \(b_j\geq0\).

For every admissible \(t\), let \(\mathcal A_t\) be the attack obtained
by replacing, only for the fixed input \(x\),
\[
K_j^x\longmapsto \sqrt{1+t\alpha_j}\,K_j^x
\]
for every \(j\). Write
\[
p(t):=\Pr_{\mathcal A_t,f}[\acc].
\]
For this input \(x\), the replacement multiplies the subnormalized
state for every outcome pair \((j,\ell)\) by \(1+t\alpha_j\).
Contributions from Alice's other inputs remain unchanged. Hence,
\[
p(t)
=
p(0)+t\sum_j\alpha_jb_j.
\]
Taking the trace of
\(\sum_j\alpha_jM_j^x=0\) gives
\[
\sum_j\alpha_j\Tr(M_j^x)=0.
\]

After discarding every outcome \(j\) for which \(M_j^x=0\), each
remaining associated POVM element \(M_j^x\) is nonzero and positive
semidefinite. Hence,
\[
\Tr(M_j^x)>0
\]
for every remaining outcome \(j\).
It follows that the nontrivial family
\(\{\alpha_j\}_j\) contains both positive and negative coefficients.
The values of \(t\) for which all weights \(1+t\alpha_j\) are
nonnegative therefore form the interval \([t_-,t_+]\), where
\[
t_-
=
\max_{j:\alpha_j>0}\left(-\frac{1}{\alpha_j}\right)<0,
\qquad
t_+
=
\min_{j:\alpha_j<0}\left(-\frac{1}{\alpha_j}\right)>0.
\]
Choose \(t=t_+\) if
\(\sum_j\alpha_jb_j\geq0\), and choose \(t=t_-\) otherwise. Then every
weight \(1+t\alpha_j\) is nonnegative, at least one of them is zero,
and
\[
p(t)-p(0)
=
t\sum_j\alpha_jb_j
\geq0.
\]
Thus, at least one nonzero outcome can be removed without decreasing
the probability of acceptance: at the chosen endpoint, there is an
outcome \(j\) for which \(1+t\alpha_j=0\), and hence its modified Kraus
operator
\[
\sqrt{1+t\alpha_j}\,K_j^x=0.
\]
Removing this zero operator leaves a valid instrument and
does not change any measurement probability. If more than \(d^2\)
nonzero outcomes remain, their associated POVM elements are again
linearly dependent, so the same argument can be applied once more.
After finitely many repetitions, at most \(d^2\) nonzero outcomes
remain for the input \(x\).

Perform this reduction independently for every Alice input \(x\), and
then for every Bob input \(y\). The commitment and split attached to
each outcome that remains are unchanged. Consequently, for every input,
Alice's resulting instrument has at most \(d^2\) nonzero outcomes, and
the same is true for Bob's resulting instrument.

Finally, let \(a_x(j)\) and \(b_y(\ell)\) identify the original
outcomes corresponding to the relabelled outcomes \(j\) and \(\ell\),
respectively. Relabel the nonzero outcomes by indices
\(j,\ell\in\{1,\ldots,d^2\}\). Each relabelled outcome keeps the
commitment, split, and classical record of its corresponding original
outcome. The final measurements are defined as the original
measurements associated with \(a_x(j)\) and \(b_y(\ell)\). Thus, the
resulting attack has the required normal form and has the same or
greater probability of acceptance.

To include countably many classical outcomes, first perform a finite
truncation. For every \(\delta>0\), choose finite sets \(J_x,L_y\) such
that
\[
\sum_{j\notin J_x}
\Tr\left[
(I_V\otimes(K_j^x)^\dagger K_j^x\otimes I_B)\rho
\right]
\leq\frac{\delta}{2},
\]
and
\[
\sum_{\ell\notin L_y}
\Tr\left[
(I_V\otimes I_A\otimes(L_\ell^y)^\dagger L_\ell^y)\rho
\right]
\leq\frac{\delta}{2}.
\]
Replace the omitted outcomes at Alice's input \(x\) by one additional
outcome
\[
K_{\rm fail}^x
=
\sqrt{\mathbb I-\sum_{j\in J_x}(K_j^x)^\dagger K_j^x},
\qquad
c_A(x,\mathrm{fail})=0^m,
\]
and analogously for Bob. This preserves normalization. Any execution
in which a failure outcome occurs is rejected: one commitment is
\(0^m\), so either the commitments disagree or their common Hamming
weight is \(0<k\). For every pair of retained
outcomes, the state and the continuation measurements are unchanged.
Consequently,
\[
0\leq p_{\rm original}-p_{\rm finite}
\leq
\frac{1}{N^2}\sum_{x,y}
\Pr[j\notin J_x\text{ or }\ell\notin L_y\mid x,y]
\leq\delta.
\]
Applying the finite-outcome reduction to this approximation and letting
\(\delta\) tend to zero covers classical records consisting of
arbitrarily long finite strings.
\end{proof}

\end{document}